\documentclass[11pt]{article}
\usepackage{custom}
\usepackage{macros}

\title{Exponential separations using guarded extension variables}
\author{
  Emre Yolcu\thanks{
    \texttt{\{eyolcu,mheule\}@cs.cmu.edu}.
    Computer Science Department,
    Carnegie Mellon University.
  }
  \and
  Marijn J.\,H. Heule\footnotemark[1]
}
\date{}

\begin{document}

\maketitle

\begin{abstract}
  We study the complexity of proof systems augmenting resolution
  with inference rules that allow, given a formula $\Gamma$ in conjunctive normal form,
  deriving clauses that are not necessarily logically implied by $\Gamma$
  but whose addition to $\Gamma$ preserves satisfiability.
  When the derived clauses are allowed to introduce variables not occurring in $\Gamma$,
  the systems we consider become equivalent to extended resolution.
  We are concerned with the versions of these systems without new variables.
  They are called $\BC^-$, $\RAT^-$, $\SBC^-$, and $\GER^-$, denoting respectively
  blocked clauses, resolution asymmetric tautologies,
  set-blocked clauses, and generalized extended resolution.
  Each of these systems formalizes some restricted version of the ability
  to make assumptions that hold ``without loss of generality,''
  which is commonly used informally to simplify or shorten proofs.

  Except for $\SBC^-$, these systems are known to be
  exponentially weaker than extended resolution.
  They are, however, all equivalent to it under a relaxed notion of simulation
  that allows the translation of the formula along with the proof
  when moving between proof systems.
  By taking advantage of this fact, we construct formulas
  that separate $\RAT^-$ from $\GER^-$ and vice versa.
  With the same strategy, we also separate $\SBC^-$ from $\RAT^-$.
  Additionally, we give polynomial-size $\SBC^-$ proofs of the pigeonhole principle,
  which separates $\SBC^-$ from $\GER^-$ by a previously known lower bound.
  These results also separate the three systems from $\BC^-$ since they all simulate it.
  We thus give an almost complete picture of their relative strengths.
\end{abstract}

\newpage

\section{Introduction}
\label{sec:introduction}

\subsection{Properties of commonly studied proof systems}
\label{sec:properties-commonly-studied-proof-systems}

Most of the commonly studied rule-based propositional proof systems,%
\footnote{Throughout this paper, by ``proof'' we mean a refutation of satisfiability
  (i.e., a derivation of $\bot$ from a set $\Gamma$ of formulas,
  where $\bot$ denotes a contradiction such as the empty clause).}
such as resolution~\cite{Bla37,Rob65}, Frege~\cite{CR79},
cutting planes~\cite{CCT87}, polynomial calculus~\cite{CEI96}, Lovász--Schrijver~\cite{LS91},
are in several senses well behaved.
For instance, they are monotonic, strongly sound,
and strongly closed under restrictions.

\begin{itemize}
\item A proof system $P$ is \emph{monotonic}
  if for all sets $\Gamma$ and $\Gamma'$ of formulas such that $\Gamma \subseteq \Gamma'$ and for every formula $\varphi$
  we have
  \begin{equation*}
    \Gamma \vdash_P \varphi \implies \Gamma' \vdash_P \varphi,
  \end{equation*}
  where $\vdash_P$ is the derivability relation for $P$.
  Since we often express a proof as a tree,
  monotonicity naturally holds for most proof systems.
  It has the consequence that the validity of an inference in the proof
  relies only on some subset of the previously derived formulas
  as opposed to the entire set.
\item \emph{Strong soundness} is the property of a proof system $P$
  that a formula $\varphi$ can be derived in $P$ from a set $\Gamma$ of formulas
  only if $\varphi$ is logically implied by $\Gamma$
  (i.e., every total assignment satisfying all formulas in $\Gamma$ also satisfies $\varphi$),
  written as
  \begin{equation}\label{eq:strong-soundness}
    \Gamma \vdash_P \varphi \implies \Gamma \models \varphi.
  \end{equation}
  Soundness is less strict than strong soundness
  in that it only requires~\cref{eq:strong-soundness} to hold for $\varphi = \bot$.
\item For a formula $\varphi$ and a partial assignment $\alpha$,
  let $\varphi|_\alpha$ denote the formula obtained by
  first replacing every assigned variable $x$ occurring in $\varphi$ by $\alpha(x)$
  and then recursively simplifying all of the subformulas.
  For a set $\Gamma$ of formulas, let $\Gamma|_\alpha \coloneqq \{\psi|_\alpha \mid \psi \in \Gamma\}$.
  We call $\varphi|_\alpha$ the restriction of $\varphi$ under $\alpha$, and similarly for $\Gamma$.
  We say a proof system $P$ is \emph{strongly closed under restrictions}
  if for every set $\Gamma$ of formulas, for every formula $\varphi$,
  and for every partial assignment $\alpha$ that does not satisfy $\varphi$, we have
  \begin{equation}\label{eq:strong-closure-under-restrictions}
    \Gamma \vdash_P \varphi \implies \Gamma|_\alpha \vdash_P \varphi|_\alpha.
  \end{equation}
  As in the case of soundness, (weak) closure under restrictions
  only requires~\cref{eq:strong-closure-under-restrictions} to hold for $\varphi = \bot$.
  Closure under restrictions is often also defined by the more quantitative condition
  that for every $P$-proof $\Pi$ of $\Gamma$ and every partial assignment $\alpha$
  there exist a $P$-proof $\Pi'$ of $\Gamma|_\alpha$ of size polynomial in the size of $\Pi$.
\end{itemize}

None of the above properties are necessary for soundness,
and proof systems that do not have them can be stronger
since such systems are more permissive
in terms of the kinds of reasoning they allow.
Possibly the most prominent example of such a system in proof complexity
is extended Frege~\cite{CR79}.
When refuting a set $\Gamma$ of formulas,
extended Frege allows
(in addition to the axioms and the rules of the underlying Frege system)
proof steps of the form
\begin{equation}\label{eq:extension-rule-frege}
  x \liff \varphi,
\end{equation}
where $x$ is a variable and $\varphi$ is an arbitrary formula,
with the condition that $x$ not occur in $\Gamma$, any of the preceding steps, or $\varphi$.
Another example is extended resolution~\cite{Tse68},
which similarly uses~\cref{eq:extension-rule-frege}
although in a more restricted form
since resolution works only with clauses.
With $x$ a ``new'' variable as before and $p$, $q$ literals,
extended resolution allows introducing $x \liff p \land q$ via the following clauses,
where the overline denotes negation:
\begin{equation}\label{eq:extension-rule-resolution}
  \lneg{x} \lor p \qquad \lneg{x} \lor q \qquad x \lor \lneg{p} \lor \lneg{q}
\end{equation}
Extended Frege (and similarly extended resolution) has none of the above properties,
although the reasons are not particularly interesting:
\begin{itemize}
\item Monotonicity fails to hold
  since we cannot necessarily derive $x \liff \varphi$ from $\Gamma' \supseteq \Gamma$
  if $x$ already occurs in $\Gamma'$.
\item Strong soundness fails to hold
  since $\Gamma$ may not imply $x \liff \varphi$ under assignments $\alpha$
  such that $\alpha(x) \neq \alpha(\varphi)$.
  Nevertheless, extended Frege is sound
  because $\Gamma$ and $\Gamma \cup \{x \liff \varphi\}$ are equisatisfiable
  (i.e., $\Gamma$ is satisfiable if and only if $\Gamma \cup \{x \liff \varphi\}$ is satisfiable),
  seen as follows:
  if an assignment $\alpha$ satisfies $\Gamma$ but falsifies $x \liff \varphi$,
  flipping $\alpha(x)$ gives a different assignment $\alpha'$ satisfying both $\Gamma$ and $x \liff \varphi$.
\item Strong closure under restrictions fails to hold
  since otherwise we could choose a partial assignment
  that only assigns $x$ to $\true$
  and conclude that every formula $\varphi$ can be derived from $\Gamma$,
  which contradicts the soundness of extended Frege.
\end{itemize}
In all of the above cases,
the counterexamples rely crucially on the fact that $x$ is a new variable.
This ability to abbreviate complex formulas by variables
significantly increases the difficulty of proving lower bounds for extended Frege
and makes it one of the strongest propositional proof systems.
(Extended resolution is equivalent to it over refutations of sets of clauses.)
From this point on, we use ``formula'' and ``set of clauses'' interchangeably.

Proof systems that violate the above properties for more sophisticated reasons
(i.e., not simply due to the introduction of new variables)
also exist.
In this paper we compare the proof complexity of four such systems
that augment resolution with inference rules of varying expressiveness.
Given a set $\Gamma$ of clauses,
these rules allow deriving clauses
that are not necessarily logically implied by $\Gamma$
but whose addition to $\Gamma$ preserves satisfiability.
We call such clauses redundant.
Deciding the redundancy of a clause
with respect to a set of clauses is $\coDP$-complete%
\footnote{The class $\DP = \{L_1 \cap L_2 \mid L_1 \in \NP, L_2 \in \coNP\}$,
  which is a superset of both $\NP$ and $\coNP$,
  was defined by Papadimitriou and Yannakakis~\cite{PY84}.}~\cite{BCB20},
so we consider only the inference rules
that rely on polynomial-time verifiable syntactic conditions
corresponding to restricted versions of redundancy.
These rules may be viewed as capturing the commonly used technique
of making assumptions that hold ``without loss of generality''
when writing informal mathematical proofs.
Such assumptions are not logically implied by the hypotheses at hand,
but their use is justified by the fact that they can be eliminated
at the possible cost of an increase in the size of the proof.
The formal rules we study rely on syntactic criteria to justify such assumptions,
with weaker criteria allowing the introduction of stronger assumptions.
In this way, these rules allow us to directly express
various kinds of informal reasoning that are otherwise difficult to formalize.

From the perspective of the broader study of proof complexity,
these systems are somewhat unique
in that Frege does not simulate even the weakest variant
unless Frege and extended Frege are equivalent~\cite[Corollary~2.5]{BT21}.
It would be interesting to determine whether some variant of these systems,
despite having the same limited syntax as resolution and no new variables,
simulates a subsystem of Frege stronger than resolution.

\subsection{Related work}
\label{sec:related-work}

\subsubsection{Proof complexity}
\label{sec:related-proof-complexity}

The inference rules we study originate from the notion of blocked clauses,
developed initially by Kullmann~\cite{Kul97,Kul99a}
to give improved deterministic algorithms for 3-SAT\@.
We call a clause $C$ blocked with respect to a set $\Gamma$ of clauses
if there exists a literal $p \in C$
such that all possible resolvents of $C$ on $p$ against clauses from $\Gamma$
are tautological (i.e., contain a literal and its negation).
Kullmann~\cite{Kul99b} showed that blocked clauses are redundant
and thus considered an inference rule that, given a set $\Gamma$ of clauses,
allows us to extend $\Gamma$ with a clause that is blocked with respect to $\Gamma$.
This rule, along with resolution,
gives the proof system called \emph{blocked clauses} $(\BC)$.
As illustrated below, $\BC$ is not monotonic, not strongly sound,
and not strongly closed under restrictions.

\begin{example}
  The clause $C = \lneg{x} \lor \lneg{y}$ is blocked
  with respect to the set $\Gamma = \{x \lor y,\ x \lor \lneg{y}\}$.
  \begin{itemize}
  \item Monotonicity fails to hold since we cannot derive $C$
    from $\Gamma' = \Gamma \cup \{y\}$ in $\BC$:
    the set $\Gamma'$ is satisfiable but $\Gamma' \cup \{C\}$ is unsatisfiable.
  \item Strong soundness fails to hold since $\Gamma$ does not imply $C$
    under assignments that set both $x$ and $y$ to $\true$.
  \item Strong closure under restrictions fails to hold
    since for an assignment $\alpha$ that sets $y$ to $\true$
    we cannot derive $C|_\alpha = \lneg{x}$ from $\Gamma|_\alpha = \{x\}$ in $\BC$:
    the set $\Gamma|_\alpha$ is satisfiable but $\Gamma|_\alpha \cup \left\{C|_\alpha\right\}$ is unsatisfiable.
  \end{itemize}
\end{example}

It is apparent from the definition of a blocked clause
that deleting clauses from $\Gamma$ enlarges the set of clauses
that are blocked with respect to $\Gamma$.
With this observation at hand, Kullmann defined a strengthening of $\BC$
called \emph{generalized extended resolution} ($\GER$)
that allows the temporary deletion of clauses from $\Gamma$.
Arbitrary deletion of clauses does not necessarily preserve satisfiability;
however, since no subset of a satisfiable $\Gamma$ is unsatisfiable,
it is also possible to further strengthen $\GER$
by allowing the arbitrary deletion of a clause as a proof step.
The resulting system is called \emph{deletion blocked clauses} ($\DBC$).

Conversely to the above point, the failure of monotonicity becomes particularly important
when deletion is not allowed since it implies that the validity
of blocked clause additions performed in sequence are order dependent.
In particular, not every set of clauses that are all blocked with respect to $\Gamma$
can be derived from $\Gamma$ by a sequence of blocked clause additions.
For this reason, proving upper bounds for generalizations of $\BC$
involves carefully ensuring the validity of sequences of inferences.

Without any additional restrictions,
the above systems all simulate extended resolution
since the clauses in~\cref{eq:extension-rule-resolution}
can be added in sequence as blocked clauses
if we are allowed to introduce new variables:
starting with a set $\Gamma$ of clauses not containing the variable $x$,
we can derive
\begin{itemize}
\item $\lneg{x} \lor p$ followed by $\lneg{x} \lor q$
  since no occurrence of the literal $x$ precedes either step
  (so both clauses are vacuously blocked), and then
\item $x \lor \lneg{p} \lor \lneg{q}$
  since its resolvents on $x$ against $\lneg{x} \lor p$ and $\lneg{x} \lor q$
  (i.e., the only preceding occurrences of $\lneg{x}$) are tautological.
\end{itemize}
The study of these systems becomes interesting when we disallow new variables.
A proof of $\Gamma$ is \emph{without new variables}
if it contains only the variables that already occur in $\Gamma$.
Throughout this paper, we denote a proof system variant
that disallows new variables with the superscript ``$-$''
(e.g., $\BC^-$ is $\BC$ without new variables).
We denote a variant that allows arbitrary deletion with the prefix $\mathsf{D}$.
All of those variants constitute examples of proof systems
that share the peculiarities of extended resolution from
\cref{sec:properties-commonly-studied-proof-systems}
without being allowed new variables.

Kullmann~\cite{Kul99b} proved that extended resolution simulates $\GER$.
He also proved that $\GER^-$ is exponentially stronger than resolution
and exponentially weaker than extended resolution.
In later work, Järvisalo, Heule, and Biere~\cite{JHB12} defined
a different generalization of $\BC$
by essentially replacing ``tautological'' in the definition of a blocked clause
with ``implied by $\Gamma$ through unit propagation.''
(Unit propagation is an automatizable but incomplete variant of resolution.)
The result is still a polynomial-time verifiable redundancy criterion
since the only important property of tautologies in the argument
for the redundancy of a blocked clause $C$ with respect to $\Gamma$
is that tautologies are implied by $\Gamma$.
This generalization is called \emph{resolution asymmetric tautologies} ($\RAT$).
Yet another generalization of $\BC$ along a different axis
is called \emph{set-blocked clauses} ($\SBC$),
defined by Kiesl, Seidl, Tompits, and Biere~\cite{KSTB18}.
We call a clause $C$ set-blocked with respect to a set $\Gamma$ of clauses
if there exists some nonempty $L \subseteq C$ such that for all $D \in \Gamma$
with $D \cap \lneg{L} \neq \varnothing$ and $D \cap L = \varnothing$
the set $\bigl(C \setminus L\bigr) \cup \bigl(D \setminus \lneg{L}\bigr)$ is tautological.
A blocked clause is the special case where $L$ is a singleton,
so set-blockedness expands the scope of the literals in $C$ that we consider.
Deciding the set-blockedness of a clause
with respect to a set of clauses is $\NP$-complete~\cite{KSTB18}.
To ensure that an $\SBC$ proof is polynomial-time verifiable,
every step in the proof that adds a clause $C$ as set-blocked
is expected to indicate the subset $L \subseteq C$ for which $C$ is set-blocked.
With that said, to reduce clutter, we leave this requirement out of our definitions
and indicate those subsets only informally throughout this paper.

Subsequent works~\cite{HKB20,KRHB20} defined further generalizations,
showed simulations between some variants, and
gave polynomial-size proofs (without new variables)
of the pigeonhole principle in a variant
called \emph{set propagation redundancy} ($\SPR^-$)
that combines $\SBC^-$ and $\RAT^-$.
Recently, Buss and Thapen~\cite{BT21} initiated a systematic study
of the proof complexity of the many generalizations of $\BC^-$.
Among other results, they showed that
the bit pigeonhole principle, parity principle, clique-coloring principle,
and Tseitin tautologies have polynomial-size $\SPR^-$ proofs.
They also showed that $\SPR^-$ can undo (with polynomial-size derivations)
the effects of or-ification, xor-ification, and lifting with index gadgets.
In view of these results, $\SPR^-$ appears to be surprisingly strong.%
\footnote{As remarked by Buss and Thapen~\cite[Section~4]{BT21},
  the apparent strength of $\SPR^-$ stems from the ability to exploit symmetries,
  which are abundant in the combinatorial principles used for
  proving lower bounds against the commonly studied proof systems.
  Other interesting examples of systems that easily prove such combinatorial principles
  are the variants of Krishnamurthy's \emph{symmetric resolution}~\cite{Kri85,Urq99,AU00,Sze05},
  obtained by augmenting resolution with rules
  that explicitly support reasoning about symmetries.}
Buss and Thapen also proved an exponential size lower bound for $\RAT^-$,
separating $\DRAT^-$ and $\SPR^-$ from it.
Superpolynomial lower bounds for $\SPR^-$ or even $\SBC^-$ are currently open.

\subsubsection{SAT solving}
\label{sec:related-SAT-solving}

As the use of SAT solvers in propositional theorem proving increased,
it became standard to expect a solver to produce a proof alongside an unsatisfiability claim.
Modern SAT solvers are based on conflict-driven clause learning (CDCL)~\cite{MS99}
and essentially search for resolution proofs.
As a result, the initial proof systems developed to help verify
the outputs of CDCL SAT solvers were based on resolution~\cite{GN03,Van08}.
However, most of the current SAT solvers go beyond CDCL
and employ an array of \emph{inprocessing techniques}~\cite{JHB12}
that transform the formula during the search.
These techniques are often not strongly sound,
and resolution falls short for expressing them.
Järvisalo, Heule, and Biere~\cite{JHB12} observed
that $\DRAT$ simulated all of the common techniques used at the time,
and, following the implementation of a practical verifier~\cite{WHH14},
$\DRAT$ became the de facto standard proof system used in SAT solvers.
Extended resolution could also be used for verification; however,
it is only known to simulate $\DRAT$ with polynomial overhead~\cite[Section~4.5]{KRHB20},
whereas $\DRAT$ simulates extended resolution with no overhead.
There are also a few examples of $\DRAT$ enabling significant gains
over the smallest known extended resolution proofs (see, e.g.,~\cite[Table~1]{KRHB20}),
which is important for practical purposes.

Another practical motivation for studying these systems is
their potential usefulness in proof search
due to the surprising strength of the variants without new variables.
Recent works have introduced a SAT solving paradigm
called \emph{satisfaction-driven clause learning} (SDCL)~\cite{HKSB17,HKB19}
that can fully automatically discover small proofs of the pigeonhole principle.
Its usefulness remains limited, though,
since it was observed to improve upon CDCL only on specific classes of formulas.
Exploiting the power of these systems might be a promising avenue
for the research that aims to improve the performance of practical SAT solvers.
To this end, it is important to understand the relative strengths of these systems.

\subsection{Results}
\label{sec:results}

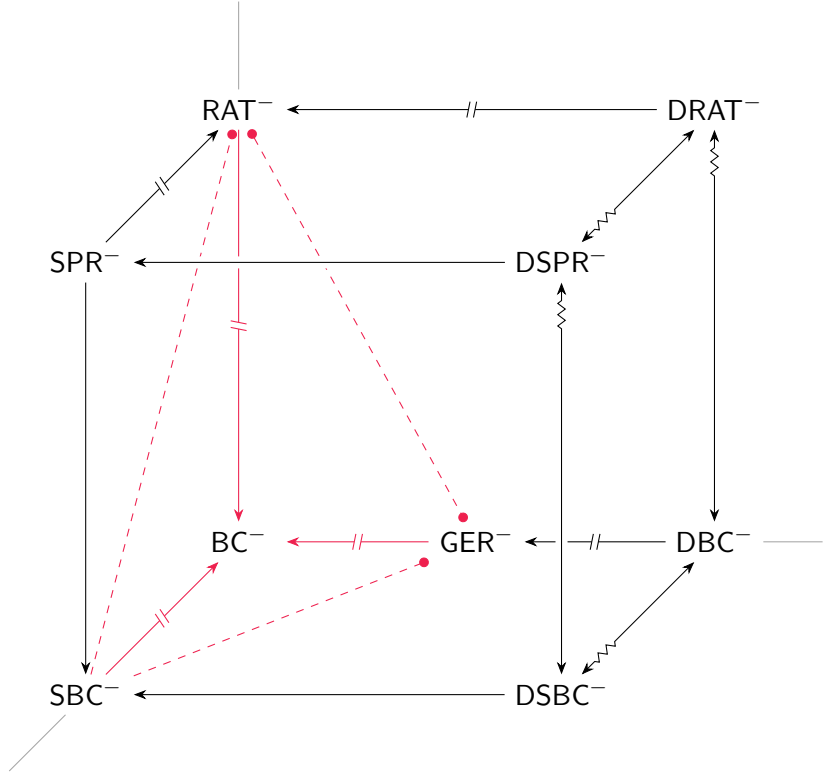
\begin{figure}[ht]
  \centering
  \begin{tikzpicture}[
      scale=2.875,
      every node/.style={minimum width={width("$\DRAT^{\smash{-}}$")}},
      axis/.style={black!35}
    ]
    \node (BC-) at (0, 0) {$\BC^{\smash{-}}$};
    \node (SBC-) at ($(BC-) + (225:1)$) {$\SBC^{\smash{-}}$};
    \node (RAT-) at (0, 2) {$\RAT^{\smash{-}}$};
    \node (SPR-) at ($(BC-) + (225:1) + (0, 2)$) {$\SPR^{\smash{-}}$};
    \node (GER-) at (1.1, 0) {$\GER^{\smash{-}}$};
    \node (DBC-) at (2.2, 0) {$\DBC^{\smash{-}}$};
    \node (DSBC-) at ($(BC-) + (225:1) + (2.2, 0)$) {$\DSBC^{\smash{-}}$};
    \node (DRAT-) at (2.2, 2) {$\DRAT^{\smash{-}}$};
    \node (DSPR-) at ($(BC-) + (225:1) + (2.2, 2)$) {$\DSPR^{\smash{-}}$};
    \draw[style=stronger] (DRAT-) -- (RAT-);
    \draw[style=stronger] (SPR-) -- (RAT-);
    \draw[style=equivalent, nontriv] (DBC-) -- (DRAT-);
    \draw[style=simulates] (SPR-) -- (SBC-);
    \draw[style=incomparable, new] (RAT-) -- (GER-);
    \draw[style=stronger, new] (RAT-) -- (BC-);
    \draw[style=stronger, new] (GER-) -- (BC-);
    \draw[style=stronger] (DBC-) -- (GER-);
    \draw[style=stronger, new] (SBC-) -- (BC-);
    \draw[style=separated, new] (SBC-) -- (GER-);
    \draw[style=separated, new] (SBC-) -- (RAT-);
    \draw[style=simulates, overdraw] (DSBC-) -- (SBC-);
    \draw[style=equivalent, overdraw, nontriv] (DBC-) -- (DSBC-);
    \draw[style=simulates, overdraw] (DSPR-) -- (SPR-);
    \draw[style=equivalent, overdraw, nontriv] (DSBC-) -- (DSPR-);
    \draw[style=equivalent, overdraw, nontriv] (DRAT-) -- (DSPR-);
    \draw[style=axis] (DBC-) -- ++(0.5, 0);
    \draw[style=axis] (RAT-) -- ++(0, 0.5);
    \draw[style=axis] (SBC-) -- ++(225:0.5);
  \end{tikzpicture}
  \caption{In the above diagram,
    the proof systems are placed in three-dimensional space with $\BC^-$ at the origin.
    Moving away from the origin along each axis
    corresponds to a particular way of generalizing a proof system.
    For systems $P$ and $Q$, we use
    $P \simulates Q$ to denote that $P$ simulates $Q$;
    (and $P \ntsimulates Q$ to indicate an ``interesting'' simulation,
    where $P$ is not simply a generalization of $Q$);
    $P \separated Q$ to denote that $P$ is exponentially separated from $Q$
    (i.e., there exists an infinite sequence of formulas
    admitting polynomial-size proofs in $P$
    while requiring exponential-size proofs in $Q$);
    and $P \stronger Q$ to denote that $P$ both simulates $Q$
    and is exponentially separated from $Q$.
    Arrows in \textcolor{customred}{red} indicate the relationships that are new in this paper.
    To reduce clutter, some relationships that are implied by transitivity are not displayed
    (e.g., $\DBC^-$ simulates $\RAT^-$ and is exponentially separated from it through $\DRAT^-$).}
  \label{fig:redundancy-landscape}
\end{figure}

We prove some results concerning
the relative strengths of $\BC^-$, $\RAT^-$, $\SBC^-$, and $\GER^-$,
continuing the line of work~\cite{Kul99b,HKB20,KRHB20,BT21}
on the proof complexity of generalizations of $\BC^-$.
\cref{fig:redundancy-landscape} summarizes the state of the proof complexity landscape
surrounding these systems after our results.

Our first main result is a two-way separation between $\RAT^-$ and $\GER^-$.

\begin{theorem}\label{thm:RAT-GER-incomparable}
  There exists an infinite sequence $(\Gamma_n)_{n=1}^\infty$ of formulas
  such that $\Gamma_n$ admits $\RAT^-$ proofs of size $n^{O(1)}$
  but requires $\GER^-$ proofs of size $2^{\Omega(n)}$.
  Conversely, there exists an infinite sequence $(\Delta_n)_{n=1}^\infty$ of formulas
  such that $\Delta_n$ admits $\GER^-$ proofs of size $n^{O(1)}$
  but requires $\RAT^-$ proofs of size $2^{\Omega(n)}$.
\end{theorem}

Since both $\RAT^-$ and $\GER^-$ are generalizations of $\BC^-$,
the above result also separates both systems from $\BC^-$.
It was already understood that $\GER^-$ is
between $\BC^-$ and $\DBC^-$ in strength,
and $\GER^-$ is in fact ``strictly'' between $\BC^-$ and $\DBC^-$
by \cref{thm:RAT-GER-incomparable}.

For both directions of the separation, we follow a strategy
that exploits the equivalence of $\BC^-$ (without new variables)
to extended resolution under \emph{effective simulations}~\cite{HHU07,PS10},
which allow the translation of the formula (in a satisfiability-preserving way)
along with the proof when moving between proof systems.
In particular, Buss and Thapen~\cite{BT21}
observed that it is possible to incorporate new variables
into a formula $\Gamma$ in such a way that $\BC^-$,
while still technically only using the variables occurring in the formula,
simulates an extended resolution proof of $\Gamma$.

\begin{lemma}[{\cite[Lemma~2.2]{BT21}}]\label{thm:BC-ER-effective-simulation}
  Suppose that a formula $\Gamma$ has an extended resolution proof of size $m$
  and that $\Gamma$ and the set $X = \{y \lor x_1 \lor \dots \lor x_m,\ y\}$ of clauses
  have no variables in common.
  Then $\Gamma \cup X$ has a $\BC^-$ proof of size $O(m)$.
\end{lemma}

To separate $\RAT^-$ and $\GER^-$, we incorporate new variables into formulas
in ways that are useful to only one of the two systems.
We achieve this by ``guarding'' the new variables by clauses
instead of providing them as in \cref{thm:BC-ER-effective-simulation}.
Recall that, for a clause to be redundant with respect to a formula
according to some syntactic criterion in this paper,
every clause in the formula has to satisfy a certain condition.
We take advantage of this fact to include the new variables
within a strategically chosen set $\cX$ of guard clauses
alongside an unsatisfiable formula $\Gamma$.
With a suitable choice of $\cX$, we are able to impose
enough limitations upon the redundant clauses derivable in a system
to ensure that the new variables can essentially be ignored in a proof of $\Gamma \cup \cX$.

For each direction of the separation,
when proving the upper bound for one of the two systems,
we show that it can efficiently work through the guard clauses
and use the new variables to simulate the extended resolution proof.
When proving the lower bound for the other system, we show essentially
that it is closed under restrictions%
\footnote{We use closure under restrictions here in the quantitative sense
  described in \cref{sec:introduction}.}
for the specific formulas
and partial assignments that we construct.
(Neither system is closed under restrictions in general.)
In other words, the guard clauses make it impossible
for the system to efficiently ``access'' the new variables,
thus preventing it from achieving any speedup.
This allows us to use the existing separations of extended resolution
from $\RAT^-$ and $\GER^-$ to separate the two systems
without needing to prove lower bounds entirely from scratch.
The main difficulty is in coming up with the appropriate ways of
incorporating new variables into formulas.

As our next main result, we separate $\SBC^-$ from $\RAT^-$
(and hence $\BC^-$) with the same strategy.
In fact, we reuse the formulas separating $\GER^-$ from $\RAT^-$
and show that $\SBC^-$ can also efficiently work through the guard clauses in them,
although in a different manner than $\GER^-$.

\begin{theorem}\label{thm:separation-SBC-from-RAT}
  There exists an infinite sequence $(\Gamma_n)_{n=1}^\infty$ of formulas
  such that $\Gamma_n$ admits $\SBC^-$ proofs of size $n^{O(1)}$
  but requires $\RAT^-$ proofs of size $2^{\Omega(n)}$.
\end{theorem}

Finally, we give polynomial-size $\SBC^-$ proofs of the pigeonhole principle,
which exponentially separates $\SBC^-$ from $\GER^-$
by a lower bound due to Kullmann~\cite[Lemma~9.4]{Kul99b}.

\begin{theorem}\label{thm:separation-SBC-from-GER}
  There exists an infinite sequence $(\Gamma_n)_{n=1}^\infty$ of formulas
  such that $\Gamma_n$ admits $\SBC^-$ proofs of size $n^{O(1)}$
  but requires $\GER^-$ proofs of size $2^{\Omega(n)}$.
\end{theorem}

Along the way to our main results, we prove a partial simulation of $\RAT^-$ by $\BC^-$.
It is partial in the sense that the size of the produced $\BC^-$ proof
is not always a polynomial in the size of the $\RAT^-$ proof
(which is impossible due to \cref{thm:RAT-GER-incomparable}).
It also has the property that, although the produced proof may sometimes be small,
the simulation cannot necessarily be carried out
in time polynomial in the size of the produced proof.
This is because the simulation involves generating satisfying assignments
to certain formulas obtained in the process.
To our knowledge, all of the ``natural'' simulations
between the commonly studied proof systems are efficient,
so the partial simulation of $\RAT^-$ by $\BC^-$ is an odd example.
Another notable aspect of the simulation
is that it directly informed the construction of the formulas
that we use for separating $\RAT^-$ from $\GER^-$.
We discuss this further at the end of \cref{sec:partial-simulation-RAT-by-BC}.
Due to the technical nature of the simulation, we do not state it here in detail.

\subsection{Open questions}
\label{sec:open-questions}

We leave open the following.

\begin{question}\label{qn:RAT-separated-from-SBC}
  Is $\RAT^-$ exponentially separated from $\SBC^-$?
\end{question}

\begin{question}\label{qn:GER-separated-from-SBC}
  Is $\GER^-$ exponentially separated from $\SBC^-$?
\end{question}

Answering these questions will complete the picture
of the relative strengths of the weakest generalizations of $\BC^-$
along each axis in \cref{fig:redundancy-landscape}.
However, we do not even have any superpolynomial lower bounds for $\SBC^-$.
If a separation of extended resolution from $\SBC^-$ is shown,
it might be possible to relatively easily separate
$\RAT^-$ and $\GER^-$ from $\SBC^-$
by tailoring guarded extension variables that $\SBC^-$ cannot access
(in the manner of the current paper).

As an aside, the formulas that we use for the separations in this paper
are arguably ``artificial'' in that they do not encode any combinatorial principles.
Separations with ``natural'' formulas give more intuitive insight
into the relative capabilities of the proof systems being considered,
so it is desirable to reprove \cref{thm:RAT-GER-incomparable,thm:separation-SBC-from-RAT}
using formulas that encode some combinatorial principles.
On the other hand, those artificial formulas enable
relatively simple and modular proofs of the separations.
An interesting open question is whether our strategy
of separating two proof systems $P$ and $Q$,
both of which effectively simulate a strong system $R$,
through syntactic manipulations of formulas
that separate $R$ from $P$ and $Q$ is more generally applicable.%
\footnote{A slightly similar strategy gives a separation of linear resolution
  from tree-like and regular resolution
  (and some other variants of resolution)~\citetext{\citealp[Lemma~4.5]{BP07}; \citealp[Section~2]{BJ16}}.}

Another desirable goal is to establish tighter connections
between the more commonly studied proof systems
and the generalizations of $\BC^-$.
As mentioned earlier, Frege does not simulate $\BC^-$
unless it also simulates extended Frege.
Additionally, it is already known
that bounded-depth Frege does not simulate $\BC^-$~\cite[Corollary~2.3]{BT21}.
We naturally wonder about the converse direction.

\begin{question}
  Is there a subsystem of Frege stronger than resolution that $\DBC^-$ simulates?
\end{question}

\section{Preliminaries}
\label{sec:preliminaries}

We denote the set of positive integers by $\NN^+$.
For $n \in \NN^+$, we let $[n] \coloneqq \{1, \dots, n\}$.
For a sequence $S = (x_1, \dots, x_n)$,
its length is $n$ and we denote it by $\abs{S}$.
We use $\langle x \rangle$ to compactly denote an infinite sequence $(x_n)_{n=1}^\infty$.

\subsection{Propositional logic}
\label{sec:propositional-logic}
For notation we mostly follow Buss and Thapen~\cite{BT21}.

We use $0$ and $1$ to denote $\false$ and $\true$, respectively.
A \emph{literal} is a propositional variable or its negation.
A set of literals is \emph{tautological}
if it contains a pair of complementary literals $x$ and $\lneg{x}$.
A \emph{clause} is the disjunction of a nontautological set of literals.
We denote by $\bV$, $\bL$, and $\bC$
respectively the sets of all variables, all literals, and all clauses.
A \emph{conjunctive normal form formula} (CNF) is a conjunction of clauses.
We identify clauses with sets of literals and CNFs with sets of clauses.
In the rest of this section we use $C$, $D$ to denote clauses
and $\Gamma$, $\Delta$ to denote CNFs.

When we know $C \cup D$ to be nontautological, we write it as $C \lor D$.
We write $C \ldor D$ to indicate a \emph{disjoint disjunction},
where $C$ and $D$ have no variables in common.
We sometimes write $\Gamma \cup \{C\}$ as $\Gamma \land C$.

We denote by $\var(\Gamma)$ the set of all the variables occurring in $\Gamma$.
We say $C$ \emph{subsumes} $D$, denoted $C \sqsupseteq D$, if $C \subseteq D$.
For CNFs, we say $\Gamma$ subsumes $\Delta$, denoted $\Gamma \sqsupseteq \Delta$,
if for all $D \in \Delta$ there exists some $C \in \Gamma$ such that $C \sqsupseteq D$.
We take the disjunction of a clause and a CNF as
\begin{equation*}
  C \lor \Delta \coloneqq \{C \lor D \mid D \in \Delta \text{ and } C \cup D \text{ is nontautological}\}.
\end{equation*}
We say $\Gamma$ and $\Delta$ are \emph{equisatisfiable}, denoted $\Gamma \eqsat \Delta$,
if they are either both satisfiable or both unsatisfiable.
With respect to $\Gamma$, a clause $C$ is \emph{redundant}
if $\Gamma \setminus \{C\} \eqsat \Gamma \eqsat \Gamma \cup \{C\}$.

A \emph{partial assignment} $\alpha$ is a partial function $\alpha \colon \bV \rightharpoonup \{0,1\}$,
which also acts on literals by letting $\alpha(\lneg{x}) \coloneqq \lneg{\alpha(x)}$.
We identify $\alpha$ with the set $\{p \in \bL \mid \alpha(p) = 1\}$,
consisting of all the literals it satisfies.
For a set $L$ of literals, we let $\lneg{L} \coloneqq \{\lneg{x} \mid x \in L\}$.
In particular, we use $\lneg{C}$ to denote the smallest partial assignment
that falsifies all the literals in $C$.
We say $\alpha$ \emph{satisfies} $C$, denoted $\alpha \models C$,
if there exists some $p \in C$ such that $\alpha(p) = 1$.
We say $\alpha$ satisfies $\Gamma$ if for all $C \in \Gamma$ we have $\alpha \models C$.
For $C$ that $\alpha$ does not satisfy,
the \emph{restriction} of $C$ under $\alpha$ is
\begin{equation*}
  C|_\alpha \coloneqq C \setminus \{p \in C \mid \alpha(p) = 0\}.
\end{equation*}
Extending the above to CNFs, the restriction of $\Gamma$ under $\alpha$ is
\begin{equation*}
  \Gamma|_\alpha \coloneqq \{C|_\alpha \mid C \in \Gamma \text{ and } \alpha \not\models C\}.
\end{equation*}

\subsection{Proof complexity}
\label{sec:proof-complexity}

We recall the definition of a proof system
(in the sense of Cook and Reckhow~\cite{CR79})
and the basic notions of proof complexity,
which can also be found in the recent textbook by Krajíček~\cite[Chapter~1]{Kra19}.

In the rest of this section we think of $\Gamma$ as a formula and $\Pi$ as a proof,
each encoded by a string over some finite alphabet.

\begin{definition}\label{def:proof-system}
  A \emph{proof system} is a polynomial-time computable binary relation $P$
  such that the following hold.
  \begin{itemize}
  \item \emph{Soundness:} For all $\Gamma$ and $\Pi$,
    if $P(\Gamma, \Pi)$ holds then $\Gamma$ is unsatisfiable.
  \item \emph{Completeness:} For all unsatisfiable $\Gamma$,
    there exists some $\Pi$ such that $P(\Gamma, \Pi)$ holds.
  \end{itemize}
  We call any $\Pi$ satisfying $P(\Gamma, \Pi)$ a \emph{$P$-proof} of $\Gamma$.
\end{definition}

Proof complexity is concerned with the sizes (or lengths) of proofs.%
\footnote{The abstract definition of proof (or formula) size
  is based on the length of the string that encodes the object;
  however, when studying concrete proof systems
  it is standard to use coarser notions of size.
  For a resolution proof, the common definition of size is the number of steps.
  It is tacitly understood that a minimum-size resolution proof
  can be encoded by a string of length polynomial in the size of the proof.}
For a proof system $P$ and a formula $\Gamma$,
we define
\begin{equation*}
  \size_P(\Gamma) \coloneqq \min \{\abs{\Pi} \mid \Pi \text{ is a $P$-proof of } \Gamma \}
\end{equation*}
if $\Gamma$ is unsatisfiable and $\size_P(\Gamma) \coloneqq \infty$ otherwise.

\begin{definition}\label{def:simulation}
  A proof system $P$ \emph{simulates} $Q$
  if for all unsatisfiable $\Gamma$ we have
  \begin{equation*}
    \size_P(\Gamma) = \size_Q(\Gamma)^{O(1)}.
  \end{equation*}
  Additionally, $P$ \emph{polynomially simulates} $Q$
  if there exists a polynomial-time algorithm
  for converting a $Q$-proof of $\Gamma$ into a $P$-proof of $\Gamma$.
\end{definition}

\begin{definition}
  Proof systems $P$ and $Q$ are \emph{equivalent}
  if they simulate each other.
  Additionally, $P$ and $Q$ are \emph{polynomially equivalent}
  if they polynomially simulate each other.
\end{definition}

We say $P$ is \emph{exponentially separated} from $Q$
if there exists some sequence $\langle \Gamma \rangle$ of formulas
such that $\size_P(\Gamma_n) = n^{O(1)}$ while $\size_Q(\Gamma_n) = 2^{\Omega(n)}$.
We call such $\langle \Gamma \rangle$ \emph{easy} for $P$ and \emph{hard} for $Q$.

\subsection{Resolution}
\label{sec:resolution}

\begin{definition}
  The \emph{resolution rule} is
  \begin{center}
    \AxiomC{$A \ldor x$}
    \AxiomC{$B \ldor \lneg{x}$}
    \BinaryInfC{$A \lor B$}
    \DisplayProof,
  \end{center}
  where $A$, $B$ are clauses and $x$ is a variable.
  We call $A \lor B$ the \emph{resolvent} of $A \lor x$ and $B \lor \lneg{x}$ on $x$.
\end{definition}

\begin{definition}
  The \emph{weakening rule} is
  \begin{center}
    \AxiomC{$A$}
    \UnaryInfC{$A \lor B$}
    \DisplayProof,
  \end{center}
  where $A$ and $B$ are clauses.
  We call $A \lor B$ a \emph{weakening} of $A$.
\end{definition}

We define a resolution proof in a slightly different form than usual:
as a sequence of CNFs instead of a sequence of clauses.

\begin{definition}
  A \emph{resolution proof} of a CNF $\Gamma$
  is a sequence $\Pi = (\Gamma_1, \dots, \Gamma_N$) of CNFs such that $\Gamma_1 = \Gamma$, $\bot \in \Gamma_N$,
  and, for all $i \in [N - 1]$, we have $\Gamma_{i+1} = \Gamma_i \cup \{C\}$,
  where
  \begin{itemize}
  \item $C$ is a resolvent of two clauses $D, E \in \Gamma_i$ or
  \item $C$ is a weakening of some clause $D \in \Gamma_i$.
  \end{itemize}
  The size of $\Pi$ is $N$.
\end{definition}

We write $\Res$ to denote the resolution proof system.
A well known fact is that resolution proofs are preserved under restrictions:
if $(\Gamma_1, \Gamma_2, \dots, \Gamma_N)$ is a resolution proof of $\Gamma$,
then, for every partial assignment $\alpha$,
the sequence $(\Gamma_1|_\alpha, \Gamma_2|_\alpha, \dots, \Gamma_N|_\alpha)$ contains a resolution proof of $\Gamma|_\alpha$.
This implies in particular the following.

\begin{lemma}\label{thm:resolution-under-restrictions}
  For every CNF $\Gamma$ and every partial assignment $\alpha$, we have
  \begin{equation*}
    \size_\Res(\Gamma|_\alpha) \leq \size_\Res(\Gamma).
  \end{equation*}
\end{lemma}

We next define a weakened version of resolution
that comes up often in the study of decision algorithms for satisfiability.

\begin{definition}
  A \emph{unit propagation proof} is a resolution proof
  where each use of the resolution rule is of the form
  \begin{center}
    \AxiomC{$A \ldor x$}
    \AxiomC{$\lneg{x}$}
    \BinaryInfC{$A$}
    \DisplayProof.
  \end{center}
\end{definition}

Unit propagation is not complete.
With $\Gamma$, $\Delta$ CNFs and $L = \{p_1, \dots, p_k\}$ a set of literals,
we define $\Gamma \land \lneg{L} \coloneqq \Gamma \land \lneg{p_1} \land \dots \land \lneg{p_k}$
and write
\begin{itemize}
\item $\Gamma \vdash_1 \bot$ to denote that there exists a unit propagation proof of $\Gamma$,
\item $\Gamma \vdash_1 L$ to denote $\left(\Gamma \land \lneg{L}\right) \vdash_1 \bot$,
\item $\Gamma \vdash_1 \Delta$ to denote that for all $D \in \Delta$ we have $\Gamma \vdash_1 D$.
\end{itemize}
Note that $\Gamma \vdash_1 \Delta$ implies $\Gamma \models \Delta$.
Moreover, whether $\Gamma \vdash_1 \Delta$ holds can be decided in polynomial time.
This makes it useful as a component in defining inference rules.

As in the case of resolution, unit propagation proofs are preserved under restrictions.

\begin{lemma}\label{thm:unit-propagation-under-restrictions}
  For every CNF $\Gamma$, every set $L$ of literals, and every partial assignment $\alpha$
  such that $\alpha \not\models L$, if $\Gamma \vdash_1 L$, then $\Gamma|_\alpha \vdash_1 L|_\alpha$.
\end{lemma}
\begin{proof}
  Suppose that $\left(\Gamma \land \lneg{L}\right) \vdash_1 \bot$.
  Since unit propagation proofs are preserved under restrictions,
  we have $\left\zerodel\left(\Gamma \land \lneg{L}\right)\right|_\alpha \vdash_1 \bot$.
  Unpacking the formula on the left-hand side gives
  \begin{equation}\label{eq:unit-propagation-under-restrictions-unpacked}
    \Gamma|_\alpha \land \left(\lneg{L} \setminus \alpha\right) \vdash_1 \bot.
  \end{equation}
  (We write $\lneg{L} \setminus \alpha$ instead of $\lneg{L}|_\alpha$
  since $\alpha$ may satisfy $\lneg{L}$.)
  Noting that $\alpha \not\models L$, we have $\lneg{L} \setminus \alpha = \lneg{L|_\alpha}$.
  Along with~\cref{eq:unit-propagation-under-restrictions-unpacked},
  this implies that $\Gamma|_\alpha \vdash_1 L|_\alpha$.
\end{proof}

From this point on, we discuss some strengthenings of the resolution proof system.

\begin{definition}\label{def:extension-clauses}
  Let $\Gamma$ be a CNF and $p$, $q$ be arbitrary literals.
  Consider a \emph{new} variable $x$
  (i.e., not occurring in any one of $\Gamma$, $p$, $q$).
  We call
  \begin{equation*}
    \{\lneg{x} \lor p,\ \lneg{x} \lor q,\ x \lor \lneg{p} \lor \lneg{q}\}
  \end{equation*}
  a set of \emph{extension clauses} for $\Gamma$.
  In this context, we refer to $x$ as the \emph{extension variable}.
\end{definition}

\begin{definition}
  A CNF $\Lambda$ is an \emph{extension} for a CNF $\Gamma$
  if there exists a sequence $(\lambda_1, \dots, \lambda_t)$ such that $\Lambda = \bigcup_{i = 1}^t \lambda_i$,
  and, for all $i \in [t]$,
  we have that $\lambda_i$ is a set of extension clauses for $\Gamma \cup \bigcup_{j = 1}^{i - 1} \lambda_j$.
\end{definition}

\begin{definition}
  An \emph{extended resolution proof} of a CNF $\Gamma$
  is a pair $(\Lambda, \Pi)$,
  where $\Lambda$ is an extension for $\Gamma$
  and $\Pi$ is a resolution proof of $\Gamma \cup \Lambda$.
  The size of $(\Lambda, \Pi)$ is defined to be $\abs{\Lambda} + \abs{\Pi}$.
\end{definition}

We write $\ER$ to denote the extended resolution proof system.

\section{Inference rules}
\label{sec:inference-rules}

We recall the redundancy criteria that lead to the inference rules
we use to augment resolution proofs.
The definitions are adapted from previous works~\cite{Kul99b,JHB12,KSTB18,HKB20,BT21}.

\begin{definition}
  A clause $C = p \ldor C'$
  is a \emph{blocked clause} (BC) for $p$ with respect to a CNF $\Gamma$
  if, for every clause $D$ of the form $\lneg{p} \ldor D'$ in $\Gamma$,
  the set $C' \cup D'$ is tautological.
\end{definition}

A strict generalization of the notion of a blocked clause
is a resolution asymmetric tautology, defined as follows.

\begin{definition}
  A clause $C = p \ldor C'$
  is a \emph{resolution asymmetric tautology} (RAT) for $p$ with respect to a CNF $\Gamma$
  if, for every clause $D$ of the form $\lneg{p} \ldor D'$ in $\Gamma$,
  we have $\Gamma \vdash_1 C' \cup D'$.
\end{definition}

Another strict generalization of a blocked clause is a set-blocked clause.%
\footnote{We define a set-blocked clause in a slightly different,
  although equivalent, way compared with the original~\cite[Definition~4.1]{KSTB18}.}

\begin{definition}
  A clause $C$ is a \emph{set-blocked clause} (SBC) for a nonempty $L \subseteq C$ with respect to a CNF $\Gamma$
  if, for every clause $D \in \Gamma$ with $D \cap \lneg{L} \neq \varnothing$ and $D \cap L = \varnothing$,
  the set $\bigl(C \setminus L\bigr) \cup \bigl(D \setminus \lneg{L}\bigr)$ is tautological.
\end{definition}

We say $C$ is a BC with respect to $\Gamma$
if there exists a literal $p \in C$
for which $C$ is a BC with respect to $\Gamma$, and similarly for RAT and SBC\@.
Note that the above definitions do not prohibit
BCs, RATs, or SBCs with respect to $\Gamma$ from containing variables not occurring in $\Gamma$.

It was shown by Kullmann~\cite{Kul99b}, Järvisalo, Heule, and Biere~\cite{JHB12},
and Kiesl, Seidl, Tompits, and Biere~\cite{KSTB18}
that BCs, RATs, and SBCs are redundant,
which makes it possible to use them to define proof systems.

\begin{theorem}\label{thm:BC-RAT-redundancy}
  If a clause $C$ is a BC, RAT, or SBC with respect to a CNF $\Gamma$,
  then $\Gamma \setminus \{C\} \eqsat \Gamma \eqsat \Gamma \cup \{C\}$.
\end{theorem}

\begin{definition}
  A \emph{blocked clauses proof} of a CNF $\Gamma$
  is a sequence $\Pi = (\Gamma_1, \dots, \Gamma_N)$ of CNFs such that $\Gamma_1 = \Gamma$, $\bot \in \Gamma_N$,
  and, for all $i \in [N - 1]$, we have $\Gamma_{i+1} = \Gamma_i \cup \{C\}$,
  where either
  \begin{itemize}
  \item $C$ is a resolvent of two clauses $D, E \in \Gamma_i$,
  \item $C$ is a weakening of some clause $D \in \Gamma_i$, or
  \item $C$ is a blocked clause with respect to $\Gamma_i$.
  \end{itemize}
  The size of $\Pi$ is $N$.
\end{definition}

We write $\BC$ to denote the blocked clauses proof system.%
\footnote{In some earlier works, $\BC$ and its generalizations are defined
  as augmentations of the \emph{reverse unit propagation} ($\RUP$) proof system~\cite{Van08} instead of resolution.
  $\RUP$ has a single inference rule,
  which allows adding a clause $C$ to a CNF $\Gamma$ if $\Gamma \vdash_1 C$.
  Resolution simulates $\RUP$ with an overhead
  linear in the number of variables (by \cref{thm:resolution-simulation-RUP}),
  and the main results in this paper are not affected by the differences in the definitions.}
Replacing ``blocked clause'' by ``resolution asymmetric tautology''
in the above definition
gives the resolution asymmetric tautologies proof system,
which we denote by $\RAT$.
Replacing it by ``set-blocked clause'' gives the set-blocked clauses proof system,
which we denote by $\SBC$.

$\RAT$ and $\SBC$ are two generalizations of $\BC$, and we now define another,
designed to overcome the dependence of the validity of $\BC$ inferences
on the order of clause additions~\cite[see][Section~1.3]{Kul99b}.

For a CNF $\Gamma$ and a set $V$ of variables, we let
\begin{equation*}
  \B_V(\Gamma) \coloneqq \{C \in \bC \mid C \text{ is a BC for a literal of some }
  x \in V \text{ with respect to } \Gamma\}.
\end{equation*}
We also let
\begin{itemize}
\item $\B(\Gamma) \coloneqq \B_\bV(\Gamma)$,
\item $\inB(\Gamma) \coloneqq \B(\Gamma) \cap \Gamma$,
\item $\B^-_V(\Gamma) \coloneqq \{C \in \B_V(\Gamma) \mid \var(C) \subseteq \var(\Gamma)\}$,
\item $\B^-(\Gamma) \coloneqq \B^-_{\var(\Gamma)}(\Gamma)$.
\end{itemize}

Before proceeding, we observe the below result,
which follows immediately from the definition of a blocked clause.

\begin{lemma}\label{thm:blocked-set-decrease}
  For all CNFs $\Gamma$ and $\Delta$ such that $\Gamma \subseteq \Delta$, we have $\B(\Gamma) \supseteq \B(\Delta)$.
\end{lemma}

Thus, we may assume without loss of generality that
all of the blocked clause additions in a $\BC$ proof
are performed before any resolution steps.
(A similar assumption does not necessarily hold for $\RAT$ proofs.)

\begin{definition}
  A sequence $(C_1, \dots, C_m)$ of some clauses from a CNF $\Gamma$
  is a \emph{maximal blocked sequence} for $\Gamma$ if
  \begin{itemize}
  \item for all $i \in [m]$ the clause $C_i$
    is blocked with respect to $\Gamma \setminus \bigcup_{j = 1}^{i - 1} \{C_j\}$ and
  \item $\inB\mleft(\Gamma \setminus \bigcup_{i = 1}^m \{C_i\}\mright)$ is empty.
  \end{itemize}
\end{definition}

For a CNF $\Gamma$, a maximal blocked sequence is unique
up to the ordering of its clauses~\cite[Lemma~6.1]{Kul99b},
which makes the following notion well defined.

\begin{definition}
  Let $(C_1, \dots, C_m)$ be a maximal blocked sequence for a CNF $\Gamma$.
  The \emph{kernel} of $\Gamma$ is
  \begin{equation*}
    \ker(\Gamma) \coloneqq \Gamma \setminus \bigcup_{i = 1}^m \{C_i\}.
  \end{equation*}
\end{definition}

\begin{definition}\label{def:blocked-extension}
  A CNF $\Lambda$ is a \emph{blocked extension} for a CNF $\Gamma$ if $\ker(\Gamma \cup \Lambda) = \ker(\Gamma)$.
\end{definition}

\begin{definition}
  A \emph{generalized extended resolution proof} of a CNF $\Gamma$
  is a pair $(\Lambda, \Pi)$,
  where $\Lambda$ is a blocked extension for $\Gamma$
  and $\Pi$ is a resolution proof of $\Gamma \cup \Lambda$.
  The size of $(\Lambda, \Pi)$ is defined to be $\abs{\Lambda} + \abs{\Pi}$.
\end{definition}

We write $\GER$ to denote the generalized extended resolution proof system.
The relationship between $\GER$ and $\BC$ is made clear
by the following characterization of blocked extensions.

\begin{lemma}[{\cite[Lemma~6.5]{Kul99b}}]\label{thm:blocked-extension-characterization}
  A CNF $\Lambda$ is a blocked extension for a CNF $\Gamma$
  if and only if there exists a CNF $\Gamma' \subseteq \Gamma$
  and an ordering $(C_1, \dots, C_m)$
  of all the clauses in $\Lambda \cup (\Gamma \setminus \Gamma')$
  such that for all $i \in [m]$
  the clause $C_i$ is blocked with respect to $\Gamma' \cup \bigcup_{j = 1}^{i - 1} \{C_j\}$.
\end{lemma}

This result gives a view of $\GER$ as a version of $\BC$
that allows the \emph{temporary} deletion of clauses from the initial formula
(i.e., clauses can be deleted as long as they are added back later).

In this paper, we study the variants of $\BC$, $\RAT$, $\SBC$, and $\GER$
that disallow the use of new variables.
We say that a proof of a CNF $\Gamma$ is \emph{without new variables}
if all the variables occurring in the proof are in $\var(\Gamma)$.
In the case of $\GER$, this constraint applies to the blocked extension.
We use $\BC^-$, $\RAT^-$, $\SBC^-$, and $\GER^-$ to denote the variants without new variables.

\subsection{Useful facts}
\label{sec:useful-facts}

We conclude this section with a few standalone results
that we will refer back to later.

\begin{lemma}\label{thm:fallback-GER-to-BC}
  For every CNF $\Gamma$ such that $\ker(\Gamma) = \Gamma$, we have
  \begin{equation*}
    \size_{\GER^-}(\Gamma) = \size_{\BC^-}(\Gamma).
  \end{equation*}
\end{lemma}
\begin{proof}
  Since $\GER^-$ is a generalization of $\BC^-$,
  the inequality $\size_{\GER^-}(\Gamma) \leq \size_{\BC^-}(\Gamma)$ immediately holds for all $\Gamma$.
  For the other direction, suppose that $\ker(\Gamma) = \Gamma$
  and that $\Lambda$ is a blocked extension for $\Gamma$.
  We claim that the clauses in $\Lambda$ can be derived in sequence from $\Gamma$ in $\BC^-$.
  By the definition of a blocked extension, we have $\ker(\Gamma \cup \Lambda) = \ker(\Gamma) = \Gamma$.
  Then, by the definition of a kernel, we can order the clauses in $\Lambda$
  into a maximal blocked sequence $(C_1, \dots, C_m)$ for $\Gamma$.
  Finally, we can derive these clauses from $\Gamma$ in the reverse order $(C_m, \dots, C_1)$
  by a sequence of blocked clause additions.
\end{proof}

\begin{lemma}\label{thm:WLOG-BC-resolution}
  For every CNF $\Gamma$, we have
  \begin{equation*}
    \size_{\BC^-}(\Gamma) \geq \size_\Res(\Gamma \cup \B^-(\Gamma)).
  \end{equation*}
\end{lemma}
\begin{proof}
  Let $\Pi$ be a $\BC^-$ proof of $\Gamma$, and let $\beta$ be the set of blocked clauses added in $\Pi$.
  Due to \cref{thm:blocked-set-decrease}, we can rearrange $\Pi$ as a resolution proof of $\Gamma \cup \beta$.
  Then, since adding clauses to a CNF cannot increase the size of its shortest resolution proof,
  and since $\beta \subseteq \B^-(\Gamma)$, the result follows.
\end{proof}

\begin{definition}\label{def:projection}
  The \emph{projection} of a CNF $\Gamma$ onto a literal $p$ is the CNF
  \begin{equation*}
    \proj_p(\Gamma) \coloneqq \{C \setminus \{p\} \mid C \in \Gamma \text{ and } p \in C\}.
  \end{equation*}
\end{definition}

This definition plays a role
in both our (partial) simulation of $\RAT^-$ by $\BC^-$
and our $\GER^-$ lower bounds.
In particular, we use the following fact,
which was already observed by Kullmann~\cite[see][Section~4]{Kul99b}.

\begin{lemma}\label{thm:BC-characterization-projection}
  A clause $C = p \ldor C'$ is a BC for $p$ with respect to a CNF $\Gamma$
  if and only if the partial assignment $\lneg{C'}$ satisfies $\proj_{\lneg{p}}(\Gamma)$.
\end{lemma}
\begin{proof}
  By the definition of a blocked clause,
  a clause $C = p \ldor C'$ is a BC for $p$ with respect to $\Gamma$
  if and only if for all $D \in \proj_{\lneg{p}}(\Gamma)$ the set $C' \cup D$ is tautological.
  Since neither of $C'$ and $D$ is tautological, $C' \cup D$ is tautological
  if and only if $\lneg{C'} \cap D \neq \varnothing$, which is equivalent to $\lneg{C'} \models D$.
\end{proof}

The next result is essentially due to Chang~\cite[Theorem~1]{Cha70}.
Although its original form is slightly weaker,
the exact statement below can be obtained
by a modification of Chang's proof.
We provide its proof in \cref{sec:proof-resolution-simulation-RUP}
for the sake of completeness.

\begin{restatable}{lemma}{resrup}\label{thm:resolution-simulation-RUP}
  For every CNF $\Gamma$ and every clause $C$ such that $\Gamma \vdash_1 C$,
  there exists a resolution derivation $(\Gamma_1, \dots, \Gamma_N)$ with $N \leq \abs{\var(\Gamma)} + 1$
  such that $\Gamma_1 = \Gamma$, $C \in \Gamma_N$, and $\Gamma \cup \{C\} \sqsupseteq \Gamma_N$.
\end{restatable}

The following gives a simple condition under which we regain monotonicity.

\begin{lemma}[{\cite[Lemma~1.20]{BT21}}]\label{thm:monotonicity}
  Let $\Gamma$ and $\Delta$ be CNFs such that $\Gamma \subseteq \Delta$ and $\Gamma \sqsupseteq \Delta$.
  If a clause is a BC, RAT, or SBC with respect to $\Gamma$,
  then it is a BC, RAT, or SBC with respect to $\Delta$.
\end{lemma}

\section{Partial simulation of \psnnv{RAT} by \psnnv{BC}}
\label{sec:partial-simulation-RAT-by-BC}

We will show how to convert a RAT addition
into a sequence of BC additions and resolution steps.
Assume that all the BCs and RATs in this section are without new variables.

\begin{definition}
  The \emph{nonblocking CNF} of a clause $C$ for a literal $p \in C$
  with respect to a CNF $\Gamma$ is
  \begin{equation*}
    \NB^\Gamma_p(C) \coloneqq \left\{D \setminus C \relmiddle|
      D \in \proj_{\lneg{p}}(\Gamma)
      \text{ and } (C \setminus \{p\}) \cup D \text{ is nontautological}\right\}.
  \end{equation*}
\end{definition}

As a consequence of the above definition,
we have $\var\mleft(\NB^\Gamma_p(C)\mright) \cap \var(C) = \varnothing$.

We say an assignment $\alpha$ \emph{minimally satisfies} a CNF $\Gamma$, denoted $\alpha \minsat \Gamma$,
if $\alpha$ satisfies $\Gamma$ while no proper subset $\alpha' \subsetneq \alpha$ satisfies $\Gamma$.
We let
\begin{equation*}
  \mu(\Gamma) \coloneqq \left\{E \in \bC \relmiddle| \lneg{E} \minsat \Gamma\right\}.
\end{equation*}
Since two different minimally satisfying assignments cannot contain one another,
no clause $E \in \mu(\Gamma)$ is contained in a different clause $E' \in \mu(\Gamma)$.

\begin{example}
  Let $\Gamma = \{x,\ y \lor \lneg{z}\}$.
  This CNF has two minimally satisfying assignments: $\{x,\, y\}$ and $\{x,\, \lneg{z}\}$.
  We thus have $\mu(\Gamma) = \{\lneg{x} \lor \lneg{y},\ \lneg{x} \lor z\}$.
\end{example}

Noting that $\Gamma \cup \mu(\Gamma)$ is unsatisfiable for every CNF $\Gamma$, we let
\begin{equation*}
  s(\Gamma) \coloneqq \abs{\mu(\Gamma)} + \size_\Res(\Gamma \cup \mu(\Gamma)).
\end{equation*}
When $\Gamma$ is unsatisfiable, we simply have $s(\Gamma) = \size_\Res(\Gamma)$.

\begin{theorem}\label{thm:simulation}
  Let $C = p \ldor C'$ be a RAT for $p$ with respect to a CNF $\Gamma$.
  There exists a $\BC^-$ derivation $(\Gamma_1, \dots, \Gamma_N)$
  such that $\Gamma_1 = \Gamma$, $C \in \Gamma_N$, and $\Gamma \cup \{C\} \sqsupseteq \Gamma_N$,
  where, letting $\Sigma = \NB^\Gamma_p(C)$ and letting $n = \abs{\var(\Gamma)}$,
  we have
  \begin{equation*}
    N \leq \abs{\Sigma} (n + 1) + s(\Sigma).
  \end{equation*}
\end{theorem}
\begin{proof}
  Since $C$ is a RAT for $p$, for all $D \in \NB^\Gamma_p(C)$ we have $\Gamma \vdash_1 C' \ldor D$,
  which implies in particular that $\Gamma \vdash_1 C \ldor D$.
  Then, using \cref{thm:resolution-simulation-RUP}, for all $D \in \NB^\Gamma_p(C)$
  we derive $C \ldor D$ from $\Gamma$ in resolution using at most $n + 1$ steps.
  More formally, we derive $\Gamma' \cup \left(C \ldor \NB^\Gamma_p(C)\right)$ from $\Gamma$,
  where $\Gamma'$ is the set of intermediate clauses,
  guaranteed by \cref{thm:resolution-simulation-RUP} to satisfy $\{C\} \sqsupseteq \Gamma'$.

  We proceed differently depending on the satisfiability of $\NB^\Gamma_p(C)$.
  \begin{description}
  \item[Case 1] (\textit{$\NB^\Gamma_p(C)$ is unsatisfiable.})
    There exists a resolution proof
    \begin{equation*}
      \Pi = (\Delta_1, \dots, \Delta_m),
    \end{equation*}
    where $\Delta_1 = \NB^\Gamma_p(C)$ and $\bot \in \Delta_m$.
    Suppose $\Pi$ is a minimum-size proof, so it does not use weakening.
    Consider the sequence
    \begin{equation*}
      \Pi' = (C \lor \Delta_1, \dots, C \lor \Delta_m).
    \end{equation*}
    This sequence is a valid resolution derivation of $C$ from $C \ldor \NB^\Gamma_p(C)$,
    seen as follows:
    \begin{itemize}
    \item By the definition of $\NB^\Gamma_p(C)$, it has no variables in common with $C$.
      Since we assumed that $\Pi$ does not use weakening,
      no subsequent CNF in $\Pi$ has any variables in common with $C$ either.

      Let $i \in [m - 1]$.
      The sequence $\Pi$ is a resolution proof, so we have $\Delta_{i + 1} = \Delta_i \cup \{E\}$,
      where $E$ is a resolvent of some $F, G \in \Delta_i$.
      Since $\Delta_i$ has no variables in common with $C$,
      it is not possible to resolve $F$ and $G$ on a variable of $C$.
      Then the clause $C \ldor E$ is a resolvent of $C \ldor F$ and $C \ldor G$,
      which are in $C \lor \Delta_{i}$.
      This proves by induction that $\Pi'$ is a valid resolution derivation.
    \item Finally, since $\bot \in \Delta_m$, we have $C \in (C \lor \Delta_m)$.
    \end{itemize}

    Thus, resolution can derive $C$ from $\Gamma$
    in at most $\abs*{\NB^\Gamma_p(C)} (n + 1) + \size_\Res\mleft(\NB^\Gamma_p(C)\mright)$ steps.
  \item[Case 2] (\textit{$\NB^\Gamma_p(C)$ is satisfiable.})
    Let $\Psi = \Gamma \cup \Gamma' \cup \left(C \ldor \NB^\Gamma_p(C)\right)$ (i.e., the current CNF).
    Since $C \sqsupseteq (\Psi \setminus \Gamma)$, the literal $\lneg{p}$ does not occur in $\Psi \setminus \Gamma$.
    Then we have $\proj_{\lneg{p}}(\Psi) = \proj_{\lneg{p}}(\Gamma)$,
    and, consequently, $\NB^\Psi_p(C) = \NB^\Gamma_p(C)$.

    Let $E$ be a clause such that $\var(E) \cap \var(C) = \varnothing$.
    By \cref{thm:BC-characterization-projection},
    the clause $p \ldor C' \ldor E$ is a BC for $p$ with respect to $\Psi$
    if and only if the partial assignment $\lneg{C' \ldor E}$ satisfies $\proj_{\lneg{p}}(\Psi)$.
    By the definition of a nonblocking CNF,
    $\lneg{C'}$ already satisfies $\proj_{\lneg{p}}(\Psi) \setminus \NB^\Psi_p(C)$,
    so $p \ldor C' \ldor E$ is a BC for $p$ with respect to $\Psi$
    if and only if $\lneg{E}$ satisfies $\NB^\Psi_p(C)$.

    Let us write $\mu$ for $\mu\mleft(\NB^\Psi_p(C)\mright)$.
    All clauses in $C \ldor \mu$ are blocked for $p$ with respect to $\Psi$,
    and the addition of each such clause of the form $C \ldor E$
    rules out, for $C \ldor \NB^\Psi_p(C)$,
    every partial assignment $\alpha$ containing $\lneg{C \cup E}$.
    Then we have
    \begin{equation*}
      \left(C \ldor \NB^\Psi_p(C)\right) \cup (C \ldor \mu) \models C,
    \end{equation*}
    where every partial assignment containing $\lneg{C}$ falsifies the left-hand side
    (i.e., $\NB^\Psi_p(C) \cup \mu$ is unsatisfiable).
    Also, no clause $E \in \mu$ contains $\lneg{p}$
    and no clause $E \in \mu$ is a subset of a different clause $E' \in \mu$.
    This implies in particular that, for every subset $\mu' \subseteq \mu$
    and for every clause $E \in \mu \setminus \mu'$,
    if $\lneg{E} \models \NB^\Psi_p(C)$,
    then $\lneg{E} \models \NB^\Psi_p(C) \cup \mu'$.
    We thus derive $C \ldor \mu$ from $\Psi$ by a sequence of blocked clause additions.

    As in the previous case, attaching $C$ to a resolution proof of $\NB^\Psi_p(C) \cup \mu$
    gives a resolution derivation of $C$,
    so we derive $C$ from $\left(C \ldor \NB^\Psi_p(C)\right) \cup (C \ldor \mu)$ in resolution.

    In the end, $\BC^-$ can derive $C$ from $\Gamma$ in at most
    \begin{equation*}
      \abs*{\NB^\Gamma_p(C)} (n + 1) + \abs{\mu} + \size_\Res\mleft(\NB^\Gamma_p(C) \cup \mu\mright)
    \end{equation*}
    steps. \qedhere
  \end{description}
\end{proof}

Given a $\RAT^-$ proof, we can apply the above theorem
to recursively replace the earliest RAT addition in the proof by a $\BC^-$ derivation.
The intermediate clauses in the derivation replacing
the addition of a RAT $C$ are all subsumed by $C$,
which ensures by \cref{thm:monotonicity}
that the validity of later RAT additions are preserved.
We can thus translate an entire $\RAT^-$ proof to a $\BC^-$ proof.

In the above simulation, when $\NB^\Gamma_p(C)$ is unsatisfiable,
we do not use any blocked clause additions.
By \cref{thm:BC-characterization-projection},
no blocked clause for $p$ exists,
and since the RAT addition by itself only gives useful information
about the clauses containing $\lneg{p}$,
a simulation where we add a clause that is blocked for a different literal
needs to be more sophisticated.
In particular, such a simulation is unlikely to be \emph{local}
in the sense of the output consisting of a sequence of derivations
that each simulate a single step in the input.
(Most simulations in proof complexity are local.)
Assuming that the above simulation is the best possible,
if every RAT addition in a $\RAT^-$ proof $\Pi$ has an unsatisfiable nonblocking CNF,
then $\BC^-$ essentially falls back to
refuting the nonblocking CNFs for locally simulating $\Pi$.
This observation hints at the transformation
in~\cref{eq:unsatisfiable-projection-transformation} for separating $\RAT^-$ from $\GER^-$.

\section{Incomparability of \psnnv{RAT} and \psnnv{GER}}
\label{sec:incomparability-RAT-and-GER}

We now show that $\RAT^-$ is exponentially separated from $\GER^-$ and vice versa,
which also exponentially separates both systems from $\BC^-$.
For both directions, we follow a strategy
similar at a high level to the one that Kullmann~\cite[Lemma~8.4]{Kul99b} used
to prove an exponential separation of $\GER^-$ from resolution.

Let $P$ and $Q$ be proof systems (without new variables) that simulate $\BC^-$.
To separate $P$ from $Q$,
we take a sequence $\langle \Gamma \rangle$ of CNFs separating $\ER$ from $Q$
and we incorporate extension variables into the formulas
in a way that allows $P$ to simulate the $\ER$ proof
while preventing $Q$ from achieving any speedup.
This strategy is made possible by the fact
that $\BC^-$ effectively simulates $\ER$.
See also the discussion by Buss and Thapen~\cite[Section~2.2]{BT21}.

From this point on, given a CNF $\Gamma$,
we use $(\Lambda^*, \Pi^*)$ to denote a minimum-size $\ER$ proof of $\Gamma$,
where $\Lambda^*$ is the union
of a sequence of $t(\Gamma) \coloneqq \abs{\Lambda^*}/3$ sets of extension clauses
such that the $i$th set $\lambda_i$ is of the form
\begin{equation*}
  \{\lneg{x_i} \lor p_i,\ \lneg{x_i} \lor q_i,\ x_i \lor \lneg{p_i} \lor \lneg{q_i}\}.
\end{equation*}
Thus, we implicitly reserve $\left\{x_1, \dots, x_{t(\Gamma)}\right\}$
as the set of extension variables used in $\Lambda^*$.
We assume without loss of generality
that the variables of $p_i$ and $q_i$ are
in $\var(\Gamma) \cup \{x_1, \dots, x_{i-1}\}$ for all $i \in [t(\Gamma)]$.

\subsection{Exponential separation of \psnnv{RAT} from \psnnv{GER}}
\label{sec:exponential-separation-RAT-from-GER}

Let $\Gamma$ be a CNF and $(\Lambda^*, \Pi^*)$ be a minimum-size $\ER$ proof of $\Gamma$ as described above.
Consider the transformation
\begin{equation}\label{eq:unsatisfiable-projection-transformation}
  \cG(\Gamma) \coloneqq \Gamma \cup \bigcup_{i = 1}^{t(\Gamma)} \bigl[(x_i \lor \Gamma) \cup (\lneg{x_i} \lor \Gamma)\bigr],
\end{equation}
where $x_1, \dots, x_{t(\Gamma)}$ are the extension variables used in $\Lambda^*$.
When $\Gamma$ is unsatisfiable,
each extension variable above is ``locked'' behind the projection $\Gamma$,
which $\RAT^-$ can overcome but $\BC^-$ cannot.

\begin{lemma}\label{thm:projection-easy-for-RAT}
  For every CNF $\Gamma$, we have
  \begin{equation*}
    \size_{\RAT^-}(\cG(\Gamma)) \leq \size_{\ER}(\Gamma).
  \end{equation*}
\end{lemma}
\begin{proof}
  We will show that the minimum-size $\ER$ proof $(\Lambda^*, \Pi^*)$ of $\Gamma$
  directly gives a $\RAT^-$ proof of $\cG(\Gamma)$ of the same size.

  We write $t$ for $t(\Gamma)$.
  Let $(\lambda_1, \dots, \lambda_t)$
  be the sequence of $t$ sets of extension clauses that make up $\Lambda^*$.
  Consider an arbitrary $i \in [t]$, and suppose that we have derived
  $\bigcup_{j = 1}^{i - 1} \lambda_j$ from $\cG(\Gamma)$ by a sequence of RAT additions,
  so the current CNF is $\Delta = \cG(\Gamma) \cup \bigcup_{j = 1}^{i - 1} \lambda_j$.
  We will introduce the clauses in
  \begin{equation*}
    \lambda_i = \{\lneg{x_i} \lor p_i,\ \lneg{x_i} \lor q_i,\ x_i \lor \lneg{p_i} \lor \lneg{q_i}\}
  \end{equation*}
  by a sequence of RAT additions.
  Note that, since $\lambda_i$ is a set of extension clauses for $\Gamma \cup \bigcup_{j = 1}^{i - 1} \lambda_j$,
  so far the variable $x_i$ occurs only in $\cG(\Gamma) \setminus \Gamma$.
  \begin{enumerate}
  \item The clause $\lneg{x_i} \lor p_i$ is a RAT for $\lneg{x_i}$ with respect to $\Delta$
    because all earlier occurrences of $x_i$
    are clauses of the form $x_i \lor D$, where $D \in \Gamma$.
    We thus require $\Delta \vdash_1 \{p_i\} \cup D$ for all $D \in \Gamma$.
    This is indeed the case since we actually have $D \in \Delta$
    by the construction of $\cG(\Gamma)$,
    which implies $\Delta \vdash_1 \{p_i\} \cup D$.
  \item The clause $\lneg{x_i} \lor q_i$ is similarly a RAT for $\lneg{x_i}$
    with respect to $\Delta \cup \{\lneg{x_i} \lor p_i\}$.
  \item The clause $x_i \lor \lneg{p_i} \lor \lneg{q_i}$
    is similarly a RAT for $x_i$ with respect to $\Delta$.
    Moreover, it is a BC for $x_i$ with respect to
    $\{\lneg{x_i} \lor p_i,\ \lneg{x_i} \lor q_i\}$
    since $\{p_i,\, \lneg{p_i},\, \lneg{q_i}\}$
    and $\{q_i,\, \lneg{p_i},\, \lneg{q_i}\}$ are both tautological.
    As a result, $x_i \lor \lneg{p_i} \lor \lneg{q_i}$ is a RAT with respect to
    $\Delta \cup \{\lneg{x_i} \lor p_i,\ \lneg{x_i} \lor q_i\}$.
  \end{enumerate}
  It follows by induction that we can derive $\Lambda^*$ from $\cG(\Gamma)$ in $\RAT^-$.
  Since $\Pi^*$ is a resolution proof of $\Gamma \cup \Lambda^*$,
  and since $\cG(\Gamma)$ contains $\Gamma$,
  we also have a resolution proof of $\cG(\Gamma) \cup \Lambda^*$.
  Thus, we have a $\RAT^-$ proof of $\cG(\Gamma)$
  of size $\abs{\Lambda^*} + \abs{\Pi^*} = \size_{\ER}(\Gamma)$.
\end{proof}

As a consequence of the above,
if a sequence $\langle \Gamma \rangle$ of CNFs is easy for $\ER$,
then, independent of whether $\langle \Gamma \rangle$ is easy or hard for $\RAT^-$,
the sequence $\cG(\langle \Gamma \rangle) \coloneqq (\cG(\Gamma_1), \cG(\Gamma_2), \dots)$ is easy for $\RAT^-$.
In contrast, the following result implies
that the extension variables added by $\cG$ are of no use to $\BC^-$.

\begin{lemma}\label{thm:projection-hard-for-BC}
  For every CNF $\Gamma$, we have
  \begin{equation*}
    \size_{\BC^-}(\cG(\Gamma)) \geq \size_\Res(\Gamma \cup \B^-(\Gamma)).
  \end{equation*}
\end{lemma}
\begin{proof}
  When $\Gamma$ is satisfiable, the inequality holds trivially,
  so suppose that $\Gamma$ is unsatisfiable.

  Applying \cref{thm:WLOG-BC-resolution} to $\cG(\Gamma)$, we have
  \begin{equation}\label{eq:projection-BC-WLOG-lower-bound}
    \size_{\BC^-}(\cG(\Gamma)) \geq \size_\Res\mleft(\cG(\Gamma) \cup \B^-(\cG(\Gamma))\mright).
  \end{equation}

  We claim that no clause in $\B^-(\cG(\Gamma))$
  is blocked for a literal of any of the variables
  in $X = \left\{x_1, \dots, x_{t(\Gamma)}\right\}$.
  To see this, consider a clause $C$ of the form $x \ldor C'$, where $x \in X$.
  If $C$ is blocked for $x$ with respect to $\cG(\Gamma)$,
  then $\lneg{C'}$ is a satisfying assignment to $\proj_{\lneg{x}}(\cG(\Gamma)) = \Gamma$
  by \cref{thm:BC-characterization-projection}.
  Since $\Gamma$ is unsatisfiable, no such assignment exists.
  Therefore, $C$ cannot be blocked for $x$, which leaves us with
  \begin{equation}\label{eq:projection-BC-no-blocked-for-extension}
    \B^-(\cG(\Gamma)) = \B^-_{\var(\Gamma)}(\cG(\Gamma)).
  \end{equation}

  Furthermore, since $\Gamma \subseteq \cG(\Gamma)$,
  every clause in $\B^-_{\var(\Gamma)}(\cG(\Gamma))$ has to be blocked
  in particular with respect to $\Gamma$.
  This requires $\B^-_{\var(\Gamma)}(\cG(\Gamma))$ to consist of clauses of the form $C \ldor D$,
  where $C \in \B^-(\Gamma)$ and $\var(D) \subseteq X$ (with $D$ possibly empty).
  In light of this, consider a partial assignment $\alpha$ such that
  \begin{equation*}
    \alpha(z) =
    \begin{cases}
      1 & \text{if } z \in X \\
      \text{undefined} & \text{otherwise}.
    \end{cases}
  \end{equation*}
  It is straightforward to see that $\cG(\Gamma)|_\alpha = \Gamma$.
  Additionally, for every clause $C \ldor D$ (of the above form) in $\B^-_{\var(\Gamma)}(\cG(\Gamma))$,
  the restriction $(C \ldor D)|_\alpha$ is either $1$ or $C$.
  We thus have
  \begin{equation}\label{eq:projection-BC-under-restriction}
    \left\zerodel\left(\cG(\Gamma) \cup \B^-_{\var(\Gamma)}(\cG(\Gamma))\right)\right|_\alpha = \Gamma \cup \B^-(\Gamma).
  \end{equation}

  Putting \cref{eq:projection-BC-WLOG-lower-bound,,%
    eq:projection-BC-no-blocked-for-extension,,%
    eq:projection-BC-under-restriction} together,
  we finally obtain
  \begin{align*}
    \size_{\BC^-}(\cG(\Gamma))
    &\geq \size_\Res\mleft(\cG(\Gamma) \cup \B^-(\cG(\Gamma))\mright) \\
    &= \size_\Res\mleft(\cG(\Gamma) \cup \B^-_{\var(\Gamma)}(\cG(\Gamma))\mright) \\
    &\geq \size_\Res\mleft(\left\zerodel\left(\cG(\Gamma)
      \cup \B^-_{\var(\Gamma)}(\cG(\Gamma))\right)\right|_\alpha\mright)
      \tag{\cref{thm:resolution-under-restrictions}} \\
    &= \size_\Res\mleft(\Gamma \cup \B^-(\Gamma)\mright),
  \end{align*}
  which is the desired inequality.
\end{proof}

For certain CNFs, the above result carries over to $\GER^-$.

\begin{lemma}\label{thm:projection-hard-for-GER}
  For every CNF $\Gamma$ such that $\ker(\Gamma) = \Gamma$, we have
  \begin{equation*}
    \size_{\GER^-}(\cG(\Gamma)) \geq \size_\Res(\Gamma \cup \B^-(\Gamma)).
  \end{equation*}
\end{lemma}
\begin{proof}
  Suppose $\ker(\Gamma) = \Gamma$.
  We will show that $\inB(\cG(\Gamma)) = \varnothing$.
  A clause $C$ in $\inB(\cG(\Gamma))$ belongs to either $\Gamma$ or $\cG(\Gamma) \setminus \Gamma$,
  and we handle the two cases separately.
  \begin{description}
  \item[Case 1] (\textit{$C \in \Gamma$.})
    If $C$ is blocked with respect to $\cG(\Gamma)$,
    then, since $\Gamma \subseteq \cG(\Gamma)$, it is blocked with respect to $\Gamma$ as well,
    which contradicts $\ker(\Gamma) = \Gamma$.
  \item[Case 2] (\textit{$C \in \cG(\Gamma) \setminus \Gamma$.})
    Since $C \in \cG(\Gamma) \setminus \Gamma$, it is of the form $p \ldor C'$,
    where $C' \in \Gamma$ and $p$ is a literal of some variable
    in $\left\{x_1, \dots, x_{t(\Gamma)}\right\}$.
    The clause $C$ is clearly not blocked for $p$ with respect to $\cG(\Gamma)$
    since there exists a clause of the form $\lneg{p} \ldor C'$ in $\cG(\Gamma)$
    but $C' \cup C'$ is not tautological.
    The clause $C$ is not blocked for any literal in $C'$ with respect to $\cG(\Gamma)$ either
    since this implies that $C'$ is blocked with respect to $\Gamma$,
    which contradicts $\ker(\Gamma) = \Gamma$.
  \end{description}
  We thus have $\ker(\cG(\Gamma)) = \cG(\Gamma)$.
  Now, using \crefnosort{thm:fallback-GER-to-BC,thm:projection-hard-for-BC}, we have
  \begin{equation*}
    \size_{\GER^-}(\cG(\Gamma)) = \size_{\BC^-}(\cG(\Gamma)) \geq \size_\Res(\Gamma \cup \B^-(\Gamma)). \qedhere
  \end{equation*}
\end{proof}

To prove the separation, we invoke the above results
with $\Gamma$ as the \emph{pigeonhole principle},
which states that every ``pigeon'' $i \in [n + 1]$ is mapped to some ``hole'' $k \in [n]$
and that no two distinct pigeons $i, j \in [n + 1]$ are mapped to the same hole.
It is defined for $n \in \NN^+$ as
\begin{equation*}
  \PHP_n \coloneqq \bigcup_{i \in [n + 1]} \left\{p_{i,1} \lor \dots \lor p_{i,n}\right\}
  \cup \bigcup_{\substack{i, j \in [n + 1],\; i \neq j \\ k \in [n]}} \{\lneg{p_{i,k}} \lor \lneg{p_{j,k}}\},
\end{equation*}
where we call the first set of clauses the \emph{pigeon axioms}
and the second set the \emph{hole axioms}.

\begin{theorem}\label{thm:separation-RAT-from-GER}
  $\RAT^-$ is exponentially separated from $\GER^-$.
\end{theorem}
\begin{proof}
  Cook~\cite{Coo76} constructed polynomial-size $\ER$ proofs of $\PHP_n$,
  which implies by \cref{thm:projection-easy-for-RAT} that
  \begin{equation*}
    \size_{\RAT^-}(\cG(\PHP_n)) = n^{O(1)}.
  \end{equation*}
  Kullmann~\cite[Theorem~2]{Kul99b} proved
  that $\size_\Res(\PHP_n \cup \B^-(\PHP_n)) = 2^{\Omega(n)}$.
  Noting that $\ker(\PHP_n) = \PHP_n$ for all $n \in \NN^+$,
  we apply \cref{thm:projection-hard-for-GER} to obtain
  \begin{equation*}
    \size_{\GER^-}(\cG(\PHP_n)) = 2^{\Omega(n)}.
  \end{equation*}
  Thus, $\cG(\langle \PHP \rangle)$ exponentially separates $\RAT^-$ from $\GER^-$.
\end{proof}

\subsection{Exponential separation of \psnnv{GER} from \psnnv{RAT}}
\label{sec:exponential-separation-GER-from-RAT}

We proceed in a similar way to the previous section.
Let $\Gamma$ be a CNF and $(\Lambda^*, \Pi^*)$ be a minimum-size $\ER$ proof of $\Gamma$.
Take a set
\begin{equation*}
  \left\{y_1, \dots, y_{t(\Gamma)}\right\} \subseteq \bV \setminus \var(\Gamma \cup \Lambda^*)
\end{equation*}
of $t(\Gamma)$ distinct variables.
Consider the transformation
\begin{equation}\label{eq:blocked-pair-transformation}
  \cH(\Gamma) \coloneqq \Gamma \cup \bigcup_{i = 1}^{t(\Gamma)} \{\lneg{x_i} \lor y_i,\ x_i \lor \lneg{y_i}\},
\end{equation}
where $x_1, \dots, x_{t(\Gamma)}$ are the extension variables used in $\Lambda^*$.
As before, only one of the two systems
can make any use of the extension variables incorporated into the formula.
This time, the temporary deletion available to $\GER^-$ makes the difference.

\begin{lemma}\label{thm:blocked-pair-easy-for-GER}
  For every CNF $\Gamma$, we have
  \begin{equation*}
    \size_{\GER^-}(\cH(\Gamma)) \leq \size_{\ER}(\Gamma).
  \end{equation*}
\end{lemma}
\begin{proof}
  Let $(\Lambda^*, \Pi^*)$ be the minimum-size $\ER$ proof of $\Gamma$.
  We will show that the clauses in $\Lambda^* \cup (\cH(\Gamma) \setminus \Gamma)$
  can be derived from $\Gamma$ in some sequence by blocked clause additions,
  which implies by \cref{thm:blocked-extension-characterization}
  that $\Lambda^*$ is a blocked extension for $\cH(\Gamma)$.

  Recall that extension clauses can be derived in sequence
  by blocked clause additions.
  Then, since $\Lambda^*$ is an extension for $\Gamma$, we derive $\Lambda^*$ by such a sequence.
  Next, from $\Gamma \cup \Lambda^*$, we derive the clauses in $\cH(\Gamma) \setminus \Gamma$.
  Let us write $t$ for $t(\Gamma)$.
  Consider the sequence
  \begin{equation*}
    (\lneg{x_1} \lor y_1,\ \dots,\ \lneg{x_t} \lor y_t,\
    x_1 \lor \lneg{y_1},\ \dots,\ x_t \lor \lneg{y_t}).
  \end{equation*}
  For each $i \in [t]$, the $i$th clause $\lneg{x_i} \lor y_i$
  in the first half of the sequence
  is blocked for $y_i$ since $\lneg{y_i}$ does not occur in any of the earlier clauses.
  Similarly, the $i$th clause $x_i \lor \lneg{y_i}$
  in the second half of the sequence
  is blocked for $\lneg{y_i}$ since the only earlier occurrence of $y_i$
  is the clause $\lneg{x_i} \lor y_i$ and $\{x_i,\, \lneg{x_i}\}$ is tautological.
  As a result, $\Lambda^*$ is a blocked extension for $\cH(\Gamma)$.

  Since $\Pi^*$ is a resolution proof of $\Gamma \cup \Lambda^*$,
  and since $\cH(\Gamma)$ contains $\Gamma$,
  we also have a resolution proof of $\cH(\Gamma) \cup \Lambda^*$.
  Thus, we have a $\GER^-$ proof of $\cH(\Gamma)$
  of size $\abs{\Lambda^*} + \abs{\Pi^*} = \size_\ER(\Gamma)$.
\end{proof}

\begin{lemma}\label{thm:blocked-pair-hard-for-RAT}
  Let $\Gamma$ be a CNF, and let $n = \abs{\var(\Gamma)}$. We have
  \begin{equation*}
    \size_{\RAT^-}(\cH(\Gamma)) \geq \frac{\size_{\RAT^-}(\Gamma)}{n + 1}.
  \end{equation*}
\end{lemma}
\begin{proof}
  We write $t$ for $t(\Gamma)$.
  Let $V = \var(\cH(\Gamma) \setminus \Gamma)$ (i.e., the set of variables added by $\cH$),
  and let $\alpha$ be a partial assignment such that
  \begin{equation*}
    \alpha(z) =
    \begin{cases}
      1 & \text{if } z \in V \\
      \text{undefined} & \text{otherwise}.
    \end{cases}
  \end{equation*}
  We claim that, for all $N \in \NN^+$, given a $\RAT^-$ derivation
  \begin{equation*}
    \Pi = (\Delta_1, \dots, \Delta_N)
  \end{equation*}
  with $\Delta_1 = \cH(\Gamma)$,
  there exists a $\RAT^-$ derivation
  \begin{equation*}
    \Pi' = (\Psi_1, \dots, \Psi_{N'})
  \end{equation*}
  with $N' \leq N \cdot (\abs{\var(\Gamma)} + 1)$
  such that
  \begin{itemize}
  \item $\Delta_1|_\alpha = \Gamma = \Psi_1$,
  \item $\Delta_N|_\alpha \subseteq \Psi_{N'}$, and
  \item $\Delta_N|_\alpha \sqsupseteq \Psi_{N'}$.
  \end{itemize}
  The second condition above implies in particular that if $\bot \in \Delta_N$, then $\bot \in \Psi_{N'}$.
  This in turn implies the desired statement,
  since it means that if $\cH(\Gamma)$ has a $\RAT^-$ proof of size $N$,
  then $\Gamma$ has a $\RAT^-$ proof of size at most $N \cdot (\abs{\var(\Gamma)} + 1)$.

  We proceed by induction.
  For a derivation $\Pi = (\Delta_1)$ of size $1$ with $\Delta_1 = \cH(\Gamma)$,
  the derivation $\Pi' = (\Delta_1|_\alpha)$ satisfies the conditions above.
  Let $\Pi = (\Delta_1, \dots, \Delta_m)$ be a $\RAT^-$ derivation with $\Delta_1 = \cH(\Gamma)$.
  Suppose that $\Pi' = (\Psi_1, \dots, \Psi_{m'})$ is a $\RAT^-$ derivation
  with $m' \leq m \cdot (\abs{\var(\Gamma)} + 1)$ satisfying the above conditions.
  Let $C$ be a clause that is derived from $\Delta_m$
  either by resolution, weakening, or RAT addition.
  We will show that there exists a $\RAT^-$ derivation from $\Psi_{m'}|_\alpha$
  of a CNF that contains and is subsumed by $(\Delta_m \cup \{C\})|_\alpha$.
  For simplicity, we will establish this for $\Delta_m|_\alpha$ instead of $\Psi_{m'}|_\alpha$,
  with the understanding that the containment and the subsumption conditions above
  imply by \cref{thm:monotonicity} that the same derivation
  can be made also from $\Psi_{m'}|_\alpha$.
  There exists a trivial derivation when $\alpha \models C$
  since it implies that $(\Delta_m \cup \{C\})|_\alpha = \Delta_m|_\alpha$,
  so suppose that $\alpha \not\models C$.
  As a consequence, $\alpha$ does not satisfy any subset of $C$.

  From this point on, we write $\Delta$ instead of $\Delta_m$ to reduce clutter.
  \begin{description}
  \item[Case 1] (\textit{$C$ is a resolvent of $D, E \in \Delta$ on $v$.})
    Without loss of generality, suppose that $v \in D$ and $\lneg{v} \in E$.
    \begin{description}
    \item[Case 1.1] (\textit{$v \in V$.})
      We have $E|_\alpha = (E \setminus \{\lneg{v}\})|_\alpha$.
      Since $(E \setminus \{\lneg{v}\}) \subseteq C$, we have $E|_\alpha \in \Delta|_\alpha$
      and $E|_\alpha \subseteq C|_\alpha$, so $C|_\alpha$ is derived from $\Delta|_\alpha$ by weakening.
    \item[Case 1.2] (\textit{$v \notin V$.})
      Since $(D \setminus \{v\}) \subseteq C$ and $(E \setminus \{\lneg{v}\}) \subseteq C$,
      and since $\alpha$ does not set $v$,
      we have $D|_\alpha \in \Delta|_\alpha$ and $E|_\alpha \in \Delta|_\alpha$.
      Moreover, $C|_\alpha$ is a resolvent of $D|_\alpha \in \Delta|_\alpha$ and $E|_\alpha \in \Delta|_\alpha$,
      so $C|_\alpha$ is derived from $\Delta|_\alpha$ by resolution.
    \end{description}
  \item[Case 2] (\textit{$C$ is a weakening of a clause $D \in \Delta$.})
    Since $D \subseteq C$, we have $D|_\alpha \in \Delta|_\alpha$ and $D|_\alpha \subseteq C|_\alpha$,
    so $C|_\alpha$ is derived from $\Delta|_\alpha$ by weakening.
  \item[Case 3] (\textit{$C$ is a RAT for $p \in C$ with respect to $\Delta$.})
    Since we assumed $\alpha \not\models C$, there are two possibilities:
    either $\lneg{p} \in V$ or $\var(p) \notin V$.
    \begin{description}
    \item[Case 3.1] (\textit{$\lneg{p} \in V$.})
      Either $p = \lneg{x_i}$ or $p = \lneg{y_i}$ for some $i \in [t]$.
      Suppose $p = \lneg{x_i}$.
      (The case for $p = \lneg{y_i}$ is symmetric.)
      Since $C$ is a RAT for $\lneg{x_i}$ with respect to $\Delta$,
      and since $x_i \lor \lneg{y_i} \in \Delta$,
      we have $\Delta \vdash_1 (C \setminus \{\lneg{x_i}\}) \cup \{\lneg{y_i}\}$.
      By \cref{thm:unit-propagation-under-restrictions},
      we also have
      \begin{equation*}
        \Delta|_\alpha \vdash_1 \left\zerodel\bigl((C \setminus \{\lneg{x_i}\})
          \cup \{\lneg{y_i}\}\bigr)\right|_\alpha,
      \end{equation*}
      which simplifies to
      \begin{equation*}
        \Delta|_\alpha \vdash_1 C|_\alpha.
      \end{equation*}
      By \cref{thm:resolution-simulation-RUP},
      there exists a resolution derivation of $C|_\alpha$
      from $\Delta|_\alpha$ of size $\abs{\var(\Delta|_\alpha)} + 1 \leq \abs{\var(\Gamma)} + 1$.
      Moreover, the final CNF in this derivation
      is subsumed by $\Delta|_\alpha \cup {C|_\alpha} = (\Delta \cup \{C\})|_\alpha$
      as desired.
    \item[Case 3.2] (\textit{$\var(p) \notin V$.})
      The clause $C$ is of the form $C' \ldor p$ with $\var(p) \notin V$.
      Since $C$ is a RAT for $p$ with respect to $\Delta$,
      for every clause $D \in \Delta$ of the form $D' \ldor \lneg{p}$,
      we have $\Delta \vdash_1 C' \cup D'$.
      We will show that $C|_\alpha$ is a RAT with respect to $\Delta|_\alpha$.
      Every clause $D \in \Delta$ such that $\alpha \models D$ simply disappears from $\Delta|_\alpha$,
      so such clauses are irrelevant when determining
      whether $C|_\alpha$ is a RAT with respect to $\Delta|_\alpha$.
      On the other hand,
      for $D \in \Delta$ of the form $D' \ldor \lneg{p}$ such that $\alpha \not\models D$,
      we have $\Delta|_\alpha \vdash_1 (C' \cup D')|_\alpha$ by \cref{thm:unit-propagation-under-restrictions}.
      Thus, $C|_\alpha$ is a RAT for $p$ with respect to $\Delta|_\alpha$. \qedhere
    \end{description}
  \end{description}
\end{proof}

Let $n = 2^k$ for $k \in \NN^+$.
For a propositional variable $x$,
let us write $x \neq 0$ and $x \neq 1$
to denote the literals $x$ and $\lneg{x}$, respectively.
To prove the separation, we invoke the above results
with $\Gamma$ as the \emph{bit pigeonhole principle},
which states that for all $i, j \in [n + 1]$ such that $i \neq j$
the binary strings $p^i_1 \dots p^i_k$ and $p^j_1 \dots p^j_k$ are different.
It is defined for $n$ as
\begin{equation*}
  \BPHP_n \coloneqq
  \bigcup_{\substack{i, j \in [n + 1],\; i \neq j \\ (h_1, \dots, h_k) \in \{0, 1\}^k}}
  \left\{\bigvee_{\ell = 1}^k p^i_\ell \neq h_\ell \lor \bigvee_{\ell = 1}^k p^j_\ell \neq h_\ell \right\}.
\end{equation*}

\begin{theorem}\label{thm:separation-GER-from-RAT}
  $\GER^-$ is exponentially separated from $\RAT^-$.
\end{theorem}
\begin{proof}
  Buss and Thapen~\cite[Theorem~4.4]{BT21}
  gave polynomial-size proofs of $\BPHP_n$
  in $\SPR^-$, which $\ER$ simulates.
  By \cref{thm:blocked-pair-easy-for-GER}, we have
  \begin{equation*}
    \size_{\GER^-}(\cH(\BPHP_n)) = n^{O(1)}.
  \end{equation*}
  They~\cite[Theorem~5.4]{BT21} also proved that $\size_{\RAT^-}(\BPHP_n) = 2^{\Omega(n)}$.
  Applying \cref{thm:blocked-pair-hard-for-RAT} gives
  \begin{equation*}
    \size_{\RAT^-}(\cH(\BPHP_n)) = 2^{\Omega(n)}.
  \end{equation*}
  Thus, $\cH(\langle \BPHP \rangle)$ exponentially separates $\GER^-$ from $\RAT^-$.
\end{proof}

\cref{thm:RAT-GER-incomparable} follows directly from
\crefnosort{thm:separation-RAT-from-GER,thm:separation-GER-from-RAT}.

\section{Exponential separation of \psnnv{SBC} from \psnnv{RAT}}
\label{sec:exponential-separation-SBC-from-RAT}

Recall the transformation in~\cref{eq:blocked-pair-transformation},
which we used to construct formulas separating $\GER^-$ from $\RAT^-$.
We had
\begin{equation*}
  \cH(\Gamma) = \Gamma \cup \bigcup_{i = 1}^{t(\Gamma)} \{\lneg{x_i} \lor y_i,\ x_i \lor \lneg{y_i}\},
\end{equation*}
where $x_1, \dots, x_{t(\Gamma)}$ are the extension variables
used in a minimum-size $\ER$ proof $(\Lambda^*, \Pi^*)$ of $\Gamma$.
We prove below that $\SBC^-$ can use those variables
and simulate the $\ER$ proof of $\Gamma$,
so deletion is not the only way to overcome the obstacles in $\cH(\Gamma)$
that prevent $\RAT^-$ from gaining any speedup.

\begin{lemma}\label{thm:blocked-pair-easy-for-SBC}
  For every CNF $\Gamma$, we have
  \begin{equation*}
    \size_{\SBC^-}(\cH(\Gamma)) \leq 2 \cdot \size_{\ER}(\Gamma).
  \end{equation*}
\end{lemma}
\begin{proof}
  We write $t$ for $t(\Gamma)$.
  Let $(\lambda_1, \dots, \lambda_t)$
  be the sequence of $t$ sets of extension clauses that make up $\Lambda^*$.
  For each $i \in [t]$, we will first derive the clauses in
  \begin{equation*}
    \lambda_i' \coloneqq \{\lneg{x_i} \lor \lneg{y_i} \lor p_i, \
    \lneg{x_i} \lor \lneg{y_i} \lor q_i, \
    x_i \lor y_i \lor \lneg{p_i} \lor \lneg{q_i}\}
  \end{equation*}
  by a sequence of SBC additions.
  Consider an arbitrary $i \in [t]$, and suppose that we have derived
  $\bigcup_{j = 1}^{i - 1} \lambda_j'$ from $\cH(\Gamma)$ by a sequence of SBC additions,
  so the current CNF is $\Delta = \cH(\Gamma) \cup \bigcup_{j = 1}^{i - 1} \lambda_j'$.
  \begin{enumerate}
  \item The clause $E_i^1 \coloneqq \lneg{x_i} \lor \lneg{y_i} \lor p_i$ is
    an SBC for $L = \{\lneg{x_i},\, \lneg{y_i}\}$ with respect to $\Delta$
    because $\lneg{x_i} \lor y_i$ and $x_i \lor \lneg{y_i}$ are the only clauses in $\Delta$
    that intersect with $\lneg{L}$, and both of these clauses also intersect with $L$
    (i.e., there is nothing to check).
  \item The clause $E_i^2 \coloneqq \lneg{x_i} \lor \lneg{y_i} \lor q_i$ is
    similarly an SBC for $L = \{\lneg{x_i},\, \lneg{y_i}\}$ with respect to $\Delta$.
    Furthermore, we have $E_i^1 \cap \lneg{L} = \varnothing$,
    so $E_i^2$ is an SBC with respect to $\Delta \cup \left\{E_i^1\right\}$.
  \item The clause $E_i^3 \coloneqq x_i \lor y_i \lor \lneg{p_i} \lor \lneg{q_i}$ is
    similarly an SBC for $M = \{x_i,\, y_i\}$ with respect to $\Delta$.
    It is also an SBC for $M$ with respect to $\left\{E_i^1, E_i^2\right\}$
    since $\bigl(E_i^3 \setminus M\bigr) \cup \bigl(E_i^1 \setminus \lneg{M}\bigr) = \{\lneg{p_i},\, \lneg{q_i},\, p_i\}$
    and $\bigl(E_i^3 \setminus M\bigr) \cup \bigl(E_i^2 \setminus \lneg{M}\bigr) = \{\lneg{p_i},\, \lneg{q_i},\, q_i\}$
    are both tautological.
    As a result, $E_i^3$ is an SBC with respect to $\Delta \cup \left\{E_i^1, E_i^2\right\}$.
  \end{enumerate}
  It follows by induction that we can derive $\bigcup_{i = 1}^t \lambda_i'$ from $\cH(\Gamma)$ in $\SBC^-$.
  For each $i \in [t]$, resolving $E_i^1$ and $E_i^2$ against $\lneg{x_i} \lor y_i$
  and resolving $E_i^3$ against $x_i \lor \lneg{y_i}$ gives $\lambda_i$,
  thus we can derive $\Lambda^*$ from $\cH(\Gamma)$ in $\SBC^-$.
  Since $\Pi^*$ is a resolution proof of $\Gamma \cup \Lambda^*$,
  and since $\cH(\Gamma)$ contains $\Gamma$,
  we also have a resolution proof of $\cH(\Gamma) \cup \Lambda^*$.
  In the end, we have an $\SBC^-$ proof of $\cH(\Gamma)$
  of size at most $2 \abs{\Lambda^*} + \abs{\Pi^*} \leq 2 \cdot \size_{\ER}(\Gamma)$.
\end{proof}

It is now straightforward to deduce \cref{thm:separation-SBC-from-RAT}
by invoking \crefnosort{thm:blocked-pair-hard-for-RAT,thm:blocked-pair-easy-for-SBC}
with $\Gamma$ as the bit pigeonhole principle
(in the manner of the proof of \cref{thm:separation-GER-from-RAT}).

\section{Exponential separation of \psnnv{SBC} from \psnnv{GER}}
\label{sec:exponential-separation-SBC-from-GER}

We now give polynomial-size $\SBC^-$ proofs of $\PHP_n$.
More specifically, we observe that the $\SPR^-$ proofs of $\PHP_n$
constructed by Buss and Thapen~\cite[Theorem~4.3]{BT21} are in fact valid $\SBC^-$ proofs,
so the proof below closely follows theirs.
(No knowledge of $\SPR^-$ is required to follow this section.)

\begin{lemma}\label{thm:polynomial-size-SBC-PHP}
  $\size_{\SBC^-}(\PHP_n) = n^{O(1)}$.
\end{lemma}
\begin{proof}
  We essentially formalize in $\SBC^-$ a short inductive proof of $\PHP_n$, which assumes
  without loss of generality that the smallest pigeon is mapped to the smallest hole,
  derives from this assumption a renamed instance of $\PHP_{n-1}$,
  and inductively repeats these steps until deriving a trivial contradiction.

  Recall that $\PHP_n$ consists of the clauses
  $P_i \coloneqq p_{i,1} \lor \dots \lor p_{i,n}$ and
  $H_{i,j,k} \coloneqq \lneg{p_{i,k}} \lor \lneg{p_{j,k}}$
  for $i, j \in [n+1]$ and $k \in [n]$ with $i \neq j$.

  For $i \in [n - 1]$, $j \in [n + 1]$, and $k \in [n]$ such that $j, k > i$, let
  \begin{equation*}
    C_{i,j,k} \coloneqq \lneg{p_{i,k}}
    \lor \lneg{p_{j,i}}
    \lor \left(\bigvee_{\substack{\ell \in [n + 1] \\ \ell \neq i}} p_{\ell,k}\right)
    \lor \left(\bigvee_{\substack{\ell \in [n + 1] \\ \ell \neq j}} p_{\ell,i}\right).
  \end{equation*}
  Also, for $i \in [n - 1]$, let
  \begin{equation*}
    \Lambda_i \coloneqq \{C_{i,j,k} \mid j \in [n + 1],\ k \in [n], \text{ and } j, k > i\}.
  \end{equation*}
  We will first show that, for all $i \in [n - 1]$,
  from $\Psi_{i-1} \coloneqq \PHP_n \cup \bigcup_{\ell = 1}^{i - 1} \Lambda_\ell$
  we can derive the clauses in $\Lambda_i$ in any order
  by a sequence of set-blocked clause additions,
  which implies that we can obtain $\Psi_{n-1}$ from $\PHP_n$ in $\SBC^-$.
  Afterwards, we will give a polynomial-size resolution derivation
  of the empty clause from $\Psi_{n-1}$, concluding the proof.

  For every clause $C_{i,j,k} \in \Lambda_i$ and every subset $\Lambda_i' \subseteq \Lambda_i \setminus \{C_{i,j,k}\}$,
  we claim that $C_{i,j,k}$ is an SBC for
  $L = \{\lneg{p_{i,k}},\, \lneg{p_{j,i}},\, p_{i,i},\, p_{j,k}\} \subseteq C_{i,j,k}$
  with respect to $\Psi_{i-1} \cup \Lambda_i'$.
  This requires us to show that for all $D \in \Psi_{i-1} \cup \Lambda_i'$
  with $D \cap \lneg{L} \neq \varnothing$ and $D \cap L = \varnothing$
  the set $\bigl(C_{i,j,k} \setminus L\bigr) \cup \bigl(D \setminus \lneg{L}\bigr)$ is tautological.
  There are three cases.

  \begin{description}
  \item[Case 1] (\textit{$D \in \PHP_n$.})
    If $D$ is a clause in $\PHP_n$ such that
    $D \cap \lneg{L} \neq \varnothing$ and $D \cap L = \varnothing$,
    then either $D = H_{i,i',i}$ for some $i' \in [n + 1]$ such that $i' \neq j$
    or $D = H_{j,j',k}$ for some $j' \in [n + 1]$ such that $j' \neq i$.
    If $D = H_{i,i',i}$, then we have $p_{i',i} \in C_{i,j,k} \setminus L$
    and $\lneg{p_{i',i}} \in D \setminus \lneg{L}$,
    so the union of the two sets is tautological.
    The argument for the case of $D = H_{j,j',k}$ is similar.
  \item[Case 2] (\textit{$D \in \bigcup_{\ell = 1}^{i - 1} \Lambda_\ell$.})
    Let $D = C_{i',j',k'}$ be an arbitrary clause in $\bigcup_{\ell = 1}^{i - 1} \Lambda_\ell$,
    where $i' < i$ and $j', k' > i'$.
    A simple inspection shows that we have
    $D \cap \lneg{L} \neq \varnothing$ if and only if $k' = i$ or $k' = k$.
    If $k' = i$, then $p_{i,i} \in D$.
    Noting that $i' \neq j$, if $k' = k$, then $p_{j,k} \in D$.
    Either way, $D \cap L = \varnothing$ fails to hold,
    so there exists no $D \in \bigcup_{\ell = 1}^{i - 1} \Lambda_\ell$
    such that $D \cap \lneg{L} \neq \varnothing$ and $D \cap L = \varnothing$.
  \item[Case 3] (\textit{$D \in \Lambda_i'$.})
    Let $D = C_{i,j',k'}$ be an arbitrary clause in $\Lambda_i'$, where $j', k' > i$.
    It is straightforward to see that we have $p_{i,i} \in D$ unless $j' = i$.
    Then, since $j' > i$, there exists no $D \in \Lambda_i'$ such that $D \cap L = \varnothing$.
  \end{description}

  Thus, we obtain $\Psi_{n-1}$ from $\PHP_n$ in $\SBC^-$,
  and now we construct a resolution proof of $\Psi_{n-1}$.
  First, for each $(i,j,k)$ such that $C_{i,j,k}$ is defined,
  we resolve $C_{i,j,k}$ against the axioms for the holes $k$ and $i$
  to derive $\lneg{p_{i,k}} \lor \lneg{p_{j,i}}$.
  Then, for each $i \in [n]$ and $j \in [n + 1]$ such that $j > i$,
  we derive the clause $\lneg{p_{j,i}}$ by induction on $i$ as follows:
  Fix $i$ and $j$ such that $j > i$. Let
  \begin{equation*}
    \Delta_{i,j} \coloneqq
    \{\lneg{p_{i,k}} \mid k \in [n] \text{ and } k < i\}
    \cup \{\lneg{p_{i,i}} \lor \lneg{p_{j,i}}\}
    \cup \{\lneg{p_{i,k}} \lor \lneg{p_{j,i}} \mid k \in [n] \text{ and } k > i\},
  \end{equation*}
  where we have the clauses in the first set from the induction hypothesis,
  the second set from $\PHP_n$, and the third set from resolving $C_{i,j,k}$
  against the hole axioms in the previous step.
  Resolving each clause in $\Delta_{i,j}$ against
  the pigeon axiom $P_i$ thus gives $\lneg{p_{j,i}}$.
  Finally, we resolve for each $i \in [n]$ the clause $\lneg{p_{n+1,i}}$
  against $P_{n+1}$ to derive the empty clause.
\end{proof}

Kullmann~\cite[Lemma~9.4]{Kul99b} showed that $\size_{\GER^-}(\PHP_n) = 2^{\Omega(n)}$,
so \cref{thm:separation-SBC-from-GER} follows by \cref{thm:polynomial-size-SBC-PHP}.

\section*{Acknowledgments}
\label{sec:acknowledgments}

We thank Jakob Nordström and Ryan O'Donnell for useful discussions
and the ITCS reviewers for their comments.

\newcommand{\Proc}[1]{Proceedings of the \nth{#1}}
\newcommand{\STOC}[1]{\Proc{#1} Symposium on Theory of Computing (STOC)}
\newcommand{\FOCS}[1]{\Proc{#1} Symposium on Foundations of Computer Science
  (FOCS)} \newcommand{\SODA}[1]{\Proc{#1} Symposium on Discrete Algorithms
  (SODA)} \newcommand{\CCC}[1]{\Proc{#1} Computational Complexity Conference
  (CCC)} \newcommand{\ITCS}[1]{\Proc{#1} Innovations in Theoretical Computer
  Science (ITCS)} \newcommand{\ICS}[1]{\Proc{#1} Innovations in Computer
  Science (ICS)} \newcommand{\ICALP}[1]{\Proc{#1} International Colloquium on
  Automata, Languages, and Programming (ICALP)}
\newcommand{\STACS}[1]{\Proc{#1} Symposium on Theoretical Aspects of Computer
  Science (STACS)} \newcommand{\LICS}[1]{\Proc{#1} Symposium on Logic in
  Computer Science (LICS)} \newcommand{\CSLw}[1]{\Proc{#1} International
  Workshop on Computer Science Logic (CSL)} \newcommand{\CSL}[1]{\Proc{#1}
  Conference on Computer Science Logic (CSL)} \newcommand{\IJCAR}[1]{\Proc{#1}
  International Joint Conference on Automated Reasoning (IJCAR)}
\newcommand{\CADE}[1]{\Proc{#1} Conference on Automated Deduction (CADE)}
\newcommand{\SAT}[1]{\Proc{#1} International Conference on Theory and
  Applications of Satisfiability Testing (SAT)} \newcommand{\RTA}[1]{\Proc{#1}
  International Conference on Rewriting Techniques and Applications (RTA)}
\newcommand{\MFCS}[1]{\Proc{#1} International Symposium on Mathematical
  Foundations of Computer Science (MFCS)} \newcommand{\TAMC}[1]{\Proc{#1}
  International Conference on Theory and Applications of Models of Computation
  (TAMC)} \newcommand{\DLT}[1]{\Proc{#1} International Conference on
  Developments in Language Theory (DLT)} \newcommand{\TABLEAUX}[1]{\Proc{#1}
  International Conference on Automated Reasoning with Analytic Tableaux and
  Related Methods (TABLEAUX)} \newcommand{\TACAS}[1]{\Proc{#1} International
  Conference on Tools and Algorithms for the Construction and Analysis of
  Systems (TACAS)} \newcommand{\LPAR}[1]{\Proc{#1} International Conference on
  Logic for Programming, Artificial Intelligence and Reasoning (LPAR)}
\newcommand{\HVC}[1]{\Proc{#1} Haifa Verification Conference (HVC)}
\newcommand{\DAC}[1]{\Proc{#1} Design Automation Conference (DAC)}
\newcommand{\DATE}{Proceedings of the Design, Automation and Test in Europe
  Conference (DATE)} \newcommand{\ISAIM}[1]{\Proc{#1} International Symposium
  on Artificial Intelligence and Mathematics (ISAIM)}
\newcommand{\AAAI}[1]{\Proc{#1} AAAI Conference on Artificial Intelligence
  (AAAI)} \newcommand{\ICML}[1]{\Proc{#1} International Conference on Machine
  Learning (ICML)} \newcommand{\NeurIPS}[1]{\Proc{#1} Conference on Neural
  Information Processing Systems (NeurIPS)}

\appendix

\section{Proof of \texorpdfstring{\cref{thm:resolution-simulation-RUP}}{Lemma}}
\label{sec:proof-resolution-simulation-RUP}

Let us recall the statement of the result before proceeding with its proof.

\resrup*
\begin{proof}
  We assume without loss of generality that $\var(C) \subseteq \var(\Gamma)$
  since any variable of $C$ not occurring in $\Gamma$
  can be ignored in a unit propagation proof of $\Gamma \land \lneg{C}$.
  Also, if $\bot \in \Gamma$, then the statement trivially follows
  since we can derive any clause by a single use of the weakening rule,
  so suppose $\bot \notin \Gamma$.

  Assuming $\Gamma \vdash_1 C$, we will show that there exists
  a resolution derivation that is as desired and,
  additionally, is of a particular shape.
  Specifically, the derivation will start by a possible use of the weakening rule,
  followed only by uses of the resolution rule
  such that one of the two premises is a clause in $\Gamma$.
  This is called an \emph{input resolution proof}.
  \cref{fig:input-resolution} shows the shape of an input resolution proof viewed as a tree.

  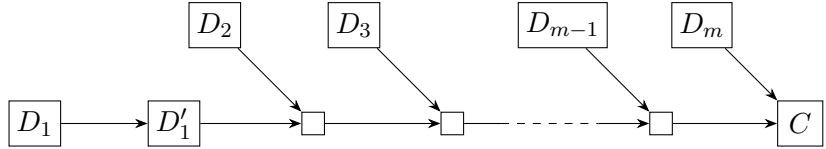
\begin{figure}[ht]
    \centering
    \begin{tikzpicture}[
        scale=1.85,
        every node/.style={draw=black, minimum size=6mm, inner sep=1mm},
        empty/.style={minimum size=3mm},
        rule/.style={->, >=Stealth}
      ]
      \node (D1) at (0, 0) {$D_1$};
      \node (D1') at (1, 0) {$D_1'$};
      \node[empty] (D2') at (2, 0) {};
      \node (D2) at ($(D2') + (135:1)$) {$D_2$};
      \node[empty] (D3') at (3, 0) {};
      \node (D3) at ($(D3') + (135:1)$) {$D_3$};
      \node[empty] (D4') at (4.5, 0) {};
      \node (D4) at ($(D4') + (135:1)$) {$D_{m-1}$};
      \node (C) at (5.5, 0) {$C$};
      \node (D5) at ($(C) + (135:1)$) {$D_m$};

      \draw[rule] (D1) edge (D1') (D1') edge (D2') (D2') edge (D3');
      \draw (D3') -- ++(0.35, 0) coordinate (a);
      \draw[dashed] (a) -- (4.05, 0) coordinate (b);
      \draw[rule] (b) -- (D4');
      \draw[rule] (D4') -- (C);
      \draw[rule] (D2) -- (D2');
      \draw[rule] (D3) -- (D3');
      \draw[rule] (D4) -- (D4');
      \draw[rule] (D5) -- (C);
    \end{tikzpicture}
    \caption{In the above tree for a resolution proof,
      each node corresponds to a clause, with incoming arrows
      indicating the premises from which it is derived.
      When we have $D_1, \dots, D_m \in \Gamma$ and $D_1 \subseteq D_1'$,
      the tree represents an input resolution proof that derives $C$ from $\Gamma$.}
    \label{fig:input-resolution}
  \end{figure}

  We proceed by induction on the number of variables of $\Gamma$.
  Let $\Gamma$ be a CNF with a single variable, say $x$.
  Let $C$ be a clause, and suppose that $\Gamma \vdash_1 C$.
  Then the clause $C$ is either $x$, $\lneg{x}$, or $\bot$.
  In the first two cases, $\Gamma \vdash_1 C$ implies that $C \in \Gamma$,
  so the derivation $(\Gamma)$ is as desired.
  In the last case, $\Gamma \vdash_1 \bot$ implies that $\Gamma = \{x,\ \lneg{x}\}$,
  so a single use of the resolution rule
  gives the derivation $(\Gamma, \Gamma \cup \{\bot\})$ as desired.

  Let $n \in \NN^+$, and suppose that the statement holds
  for all formulas and clauses with $n$ variables.
  Consider a CNF $\Gamma$ with $\abs{\var(\Gamma)} = n + 1$.
  Let $C$ be a clause, and suppose that $\Gamma \vdash_1 C$.
  Since $\left(\Gamma \land \lneg{C}\right) \vdash_1 \bot$,
  the CNF $\Gamma \land \lneg{C}$ must have some clause
  containing a single literal, say $p$
  (otherwise the unit propagation rule cannot be used at all).
  Let $\alpha$ be the partial assignment that only sets $\alpha(p) = 1$.
  There are two cases.

  \begin{description}
  \item[Case 1] (\textit{$\alpha \models C$.})
    We have $p \in C$, which implies that $\lneg{p} \notin C$ since $C$ is a clause.
    As a result, $p \notin \lneg{C}$, but we know that $\{p\} \in \left(\Gamma \land \lneg{C}\right)$,
    so it must be the case that $\{p\} \in \Gamma$.
    Then we can derive $C$ from $\Gamma$ by a single use of the weakening rule,
    so the derivation $(\Gamma, \Gamma \cup \{C\})$ is as desired.
  \item[Case 2] (\textit{$\alpha \not\models C$.})
    By \cref{thm:unit-propagation-under-restrictions}, we have $\Gamma|_\alpha \vdash_1 C|_\alpha$.
    Since $\abs{\var(\Gamma|_\alpha)} = n$, the induction hypothesis
    guarantees the existence of an input resolution derivation
    $\Pi = (\Gamma_1, \dots, \Gamma_N)$ with $N \leq n + 1$
    such that $\Gamma_1 = \Gamma|_\alpha$, $C|_\alpha \in \Gamma_N$, and $\Gamma|_\alpha \cup \{C|_\alpha\} \sqsupseteq \Gamma_N$.
    For every leaf $D$ in the tree for $\Pi$, either $D \in \Gamma$ or $D \lor \lneg{p} \in \Gamma$,
    so we define
    \begin{equation*}
      f(D) \coloneqq
      \begin{cases}
        D & \text{if } D \in \Gamma \\
        D \lor \lneg{p} & \text{if } D \lor \lneg{p} \in \Gamma.
      \end{cases}
    \end{equation*}
    We modify $\Pi$ by first replacing each leaf $D$ by $f(D)$
    and then recursively replacing each inner node $E$ by $E \lor \lneg{p}$
    if one of its premises has been changed.
    The result corresponds to a valid input resolution derivation $\Pi' = (\Gamma_1', \dots, \Gamma_N')$
    such that $\Gamma_1' = \Gamma$, $C' \in \Gamma_N'$, and $\Gamma \cup \{C|_\alpha\} \sqsupseteq \Gamma_N'$,
    where either $C' = C|_\alpha$ or $C' = C|_\alpha \lor \lneg{p}$.
    Moreover, we have either $C|_\alpha = C$ or $C|_\alpha = C \setminus \{\lneg{p}\}$.
    We make a final modification to $\Pi'$ according to the following table:
    \begin{center}
      \renewcommand{\arraystretch}{1.2}
      \setlength{\tabcolsep}{10pt}
      \begin{tabular}{l p{12em} p{12em}}
        \toprule
        & $C' = C|_\alpha$ & $C' = C|_\alpha \lor \lneg{p}$ \\
        \midrule
        $C|_\alpha = C$ & Nothing to do, because $\Pi'$ is already as desired.
                      & Resolve $C'$ against $\{p\} \in \Gamma$,
                        where $\{p\} \in \Gamma$ holds due to $\lneg{p} \notin C$
                        and $\{p\} \in \left(\Gamma \land \lneg{C}\right)$. \\
        $C|_\alpha = C \setminus \{\lneg{p}\}$
        & \multicolumn{2}{p{24em + 2\tabcolsep}}{
          Modify $\Pi'$ to make sure that the clause derived
          by the initial use of the weakening rule
          (inserting one if weakening was not used)
          contains $\lneg{p}$, followed by the same modification
          for all of the clauses corresponding to the inner nodes
          of the proof tree.} \\
        \bottomrule
      \end{tabular}
    \end{center}

    In any case, we increase the size of $\Pi'$ by at most one
    and end up with the desired derivation
    of size at most $n + 2 = \abs{\var(\Gamma)} + 1$. \qedhere
  \end{description}
\end{proof}

\end{document}